\documentclass{article}

\usepackage[left=2cm,right=4cm,top=1.5cm,bottom=1.5cm]{geometry}
\usepackage[utf8]{inputenc}
\usepackage{amsthm}
\newtheorem{definition}{Definition}
\usepackage{svg}
\usepackage{soul}
\usepackage{amssymb, bm}

\usepackage{eurosym}
\usepackage{multirow}
\usepackage{multicol}
\usepackage{amsfonts}
\usepackage{paralist}
\usepackage{comment}
\usepackage{booktabs}
\usepackage{siunitx}
\usepackage{xcolor}

\usepackage{color, colortbl}
\definecolor{Gray}{gray}{0.9}
\definecolor{Reddish}{rgb}{0.9,0,0}
\usepackage[first=0,last=9]{lcg}

\usepackage[ruled,vlined]{algorithm2e}
\usepackage{float}
\usepackage{authblk}
\usepackage{algorithmic}
\usepackage[shortlabels]{enumitem}

\usepackage{amstext} 
\usepackage{array}   
\newcolumntype{L}{>{$}l<{$}} 

\usepackage{caption}
\usepackage{subcaption}
\usepackage{amsmath}
\usepackage{enumitem}

\usepackage{tikz}
\usepackage[draft]{todonotes}
\usepackage{amsopn}

\usepackage{carshapes}
\usepackage{graphicx}
\usepackage{pgfplots}

\usetikzlibrary{arrows.meta}
\newtheorem{theorem}{Theorem}[section]
\newtheorem{corollary}{Corollary}[theorem]
\newtheorem{lemma}[theorem]{Lemma}
\newtheorem{example}[theorem]{Example}

\usepackage[
backend=biber,
style=nature,
maxcitenames=2,
uniquelist=false,
citestyle=authoryear]{biblatex} 
\definecolor{rev}{rgb}{0.0, 0.0, 0}

\DeclareMathOperator*{\argmin}{argmin}

\title{How to split the costs among travellers sharing a ride? Aligning system's optimum with users' equilibrium}
\author[1,*]{Andres Fielbaum}
\author[2]{Rafa\l{} Kucharski}
\author[2]{Oded Cats}
\author[1]{Javier Alonso-Mora}
\affil[1]{\small Department of Cognitive Robotics, Delft University of Technology, The Netherlands}
\affil[2]{\small Department of Transport \& Planning, Delft University of Technology, The Netherlands}
\affil[*]{\small Corresponding author: a.s.fielbaumschnitzler@tudelft.nl}
\date{}

\begin{document}
\maketitle
\begin{abstract}
How to form groups in a mobility system that offers shared rides, and how to split the costs within the travellers of a group, are non-trivial tasks, as two objectives conflict: 1) minimising the total costs of the system, and 2) making each user content with her assignment. Aligning both objectives is challenging, as users are not aware of the externalities induced to the rest of the system. 

In this paper, we propose protocols to share the costs within a ride so that optimal solutions can also constitute equilibria. To do this, we model the situation as a game. We show that the traditional notions of equilibrium in game theory (Nash and Strong) are not useful here, and prove that determining whether a Strong Equilibrium exists is an NP-Complete problem. Hence, we propose three alternative equilibrium notions (stronger than Nash and weaker than Strong), depending on how users can coordinate, that effectively represent stable ways to match the users. We then propose three cost-sharing protocols, for which the optimal solutions are an equilibrium for each of the mentioned intermediate notions of equilibrium. The game we study can be seen as a game-version of the well-known \textit{set cover problem}.

Numerical simulations for Amsterdam reveal that our protocols can achieve stable solutions that are always close to the optimum, that there exists a trade-off between total users' costs and how equal do they distribute among them, and that having a central coordinator can have a large impact.





\end{abstract}


\section{Introduction}
\subsection{Aligning stakeholders' interests in shared rides}
Transport systems often face a tension between individual choices and global optimisation. As users need to share some limited resources, such as streets or vehicles, their decisions affect other travellers but such externalities are usually not internalised. There are some famous paradoxes illustrating this issue, such as the Braess paradox (\cite{frank1981braess}), which predicts that building an extra road might make everybody worse through their selfish routing decisions, or the Down-Thomson paradox (\cite{mogridge1997self}), which states that new infrastructure for cars can also lead to an overall deterioration due to users switching from public transport to private modes. In public transport, similar situations can occur with new transit lines, that might increase the number of people that cannot board a vehicle due to overcrowding (\cite{renken2018demand}), or may split the users of one line into many, leading to a general reduction in frequencies (\cite{fielbaum2020beyond}), so that a new line can have a negative general effect (\cite{jara2012public}).

The problem of aligning users' interests manifests itself in a novel way in on-demand shared systems, where requests are matched into groups, meaning that the route (thus waiting and in-vehicle times) depends on the circumstantial co-travellers. Grouping might also affect fares, as the operating costs can now be split over the members of the group. That is to say, if users can decide with whom to travel, they might induce externalities to the rest of the users, which implies that their individual interests might not be perfectly aligned with a global optimisation process.

Therefore, in such shared mobility systems it is crucial to distinguish between \textit{optima} and \textit{equilibria}. Optimal solutions are typically pursued by some central non-profit operator or authority, who wants the system to run as efficiently as possible, and refer to minimising the sum of all the costs involved. Equilibria (also called \textit{stable} solutions) deal with users' interests, namely to ensure that everybody is satisfied enough so that they will not coordinate to change the solution. As we discuss below, there are different ways to define an equilibrium depending on how users are assumed to be able to coordinate.

The authority does have a tool so that users internalise (at least to some extent) the externalities they induce to the rest of the system: fares. In fact, when several users share a vehicle, an inevitable question of how to split the monetary costs emerges, and there is no straightforward answer to it. Naive approaches might be to split the costs uniformly between the users, or proportionally to the distance between each user's origin and destination. However, such ideas might lead to undesirable equilibria. If a user is travelling a long distance from her origin to her destination, then nobody would want to share the vehicle with her under a uniform division of the costs; the opposite situation would occur with fares that are proportional to the distance (a common scheme applied in ridepooling systems), as nobody would be willing to share a ride with someone requiring a short trip. Therefore, the question of which \textit{cost-sharing protocol} should be implemented in order to align the users' and the authority's interests is far from trivial. 

This paper is primarily devoted to addressing this research gap.
To do so, we model the described situation as a formal game, in which each user can choose with whom to travel, as long as the selected co-travellers agree. We argue that the usual equilibrium notions (Nash and Strong Equilibria) do not capture appropriately how users can coordinate, which requires us to propose some alternative equilibria notions. For each of these notions, we propose a cost-sharing protocol (i.e., how to share the costs among the users within a group) that makes the optimal solution an equilibrium, so that the authority can propose the users how to match (optimally) in a way that they will be satisfied. Finally, we test our ideas using real-life data in Amsterdam, with a batch of 400 travellers sharing rides.

The equilibrium notions we describe, as well as the corresponding cost-sharing protocols, might be utilised for any mobility system in which groups are formed on-demand and where a sufficient number of vehicles is always available (i.e., there are no rejected requests). The rest of the paper is written assuming a \textit{ridepooling} system,  i.e., a centralised service that matches travellers into groups and assigns vehicles to serve them. Such systems have been intensively studied in recent years (\cite{alonso2017demand,kucharski2020exact,ota2016stars,tsao2019model,fagnant2018dynamic}), and the underlying optimisation problem has been shown to be quite difficult, as for some specifications it extends set cover (\cite{fielbaum2020OnDemand}), although some approximation results have been recently proved (\cite{mori2020request}). Ridepooling services are considered promising for the future of mobility, as they might keep many of the virtues that have made ride-hailing services popular, while reducing the increase in congestion that has been associated with those (\cite{henao2019impact,tirachini2020does,agarwal2019impact,diao2021impacts,wu2021assessing}). Aligning users' and the system's interests is crucial for this purpose, as one needs users to be interested in using the system, and the system to be able to effectively stimulate users to share so that congestion is indeed mitigated. In this paper the focus is on how to group the demand (travellers) for a specific assignment decision, whereas the operation and routing of the vehicles is assumed to be controlled by a system in response to such groups (and not explicitly modelled here).

Besides ridepooling systems, our findings also apply to \textit{ridesharing} services, in which different riders coordinate to transport in a vehicle driven by one of them (\cite{agatz2012optimization,chan2012ridesharing,furuhata2013ridesharing,mourad2019survey,enzi2021modeling,ozkan2020joint}). For instance, \cite{lu2019fair} study a similar model, arguing that when groups are matched, any of the riders within a group can be the driver. Although these systems have not been able to become large-scale yet, they might become a relevant part of future shared mobility systems, when coordination tools continue to evolve. As these systems might operate in non-centralised fashions, users' choices become a structural component, so understanding their impact is crucial.

\subsection{An illustrative example}\label{scn:example}
Let us motivate the paper and the ideas we propose through an illustrative example, whose details will be expanded in the corresponding sections of the paper. First we present a travel demand of three travellers (Figure \ref{Img:ExampleIntro}). We list the groups that can be formed by the pooling travellers and their total costs in Table \ref{tab:ExampleIntro}, based on which we discuss the optimal global solution for the matching problem. We then look at the problem from the individual perspective by considering two possible ways to split those costs among the users also in Table \ref{tab:ExampleIntro}. The first one is arbitrary, and we discuss that there might not be any stable solution if those were the actual individual costs. The second one results from one of the cost-sharing protocols we define later on in the paper (called ``residual-based''), and we show that the optimal solution becomes also stable, that is, accepted from the users' perspective:

Three users require to be transported to a common destination, a demand that could be fulfilled with one, two or three vehicle rides. Figure \ref{Img:ExampleIntro} shows the distances between the nodes, where we mark in green the path that would be required to serve them all with a single vehicle. The second column in Table \ref{tab:ExampleIntro} shows the total cost of each group (encompassing users' and operators' costs), implying that the optimal way to serve the system (i.e., the one that minimises the total cost of the chosen groups) would be to transport users $1$ and $2$ together, and user $3$ travelling alone (marked with a grey background in Table \ref{tab:ExampleIntro}). The origin of user 2 is situated on the shortest path of user 1, which is why it is efficient to group them.

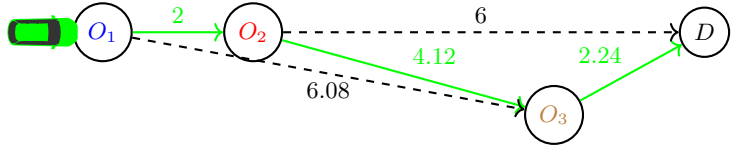
\begin{figure}[h]
    \centering
    \small
    \begin{tikzpicture}[->>=1pt,auto,node distance=2cm,
          thick,main node/.style={circle,draw}]
        \node[main node] at (-6,0) (0) {\color{blue}$O_1$};
          \node[main node] at(-4,0) (1) {\color{red}$O_2$};
         \node[main node] at(0,-1.1) (2) {\color{brown}$O_3$}; 
        \node[main node] at(2,0) (D) {$D$};
          \node [sedan top,body color=green,window color=black!80,minimum width=1cm] at (-6.8,0) {};

         \path [->,green] (0) edge node {2} (1);
         \path [->,green] (1) edge node {4.12} (2);
        \path [->,green] (2) edge node {2.24} (D);
        \path [->,dashed] (1) edge node {\color{black}6} (D);
        \path [->,dashed,below] (0) edge node {\color{black}6.08} (2);
        \end{tikzpicture}
    \caption{An illustrative example of three users that might share a ride. The origins are marked with different colours, and they have a common destination. Each arc shows its length. The green arcs represent the path that would be followed by the green vehicle if serving the whole group.}
    \label{Img:ExampleIntro}
\end{figure}

\begin{table}[H]
    \centering
    \begin{tabular}{|c|c|c|c|c|c|c|c|}
    \hline
    \multirow{2}{*}{\textbf{Group}}     & \multirow{2}{*}{\textbf{Total cost}} & \multicolumn{3}{c|}{\textbf{Cost per user 1}} & \multicolumn{3}{c|}{\textbf{Cost per user 2}}   \\
    \cline{3-8}
    & & User 1 & User 2 & User 3 & User 1 & User 2 & User 3 \\
    \hline
    $\{1\}$ & 23 & 23 & - & - & 23 & - & - \\
    $\{2\}$ & 19 & - & 19 & - & - & 19 & - \\
    \rowcolor{Gray} 
    $\{3\}$ & 9.24 & - & - & 9.24 & - & - & 9.24 \\
    \rowcolor{Gray}
    $\{1,2\}$ & 31 & 22.99 & 8.01 & - & 16.98 & 14.02 & -\\
    $\{1,3\}$ & 31.96 & 22.96 & - & 9 & 22.8 & - & 9.16 \\
    $\{2,3\}$ & 26.08 & - & 18 & 8.08 & - & 17.55 & 8.53 \\
    $\{1,2,3\}$ & 40.44 & 23 & 12 & 5.44 & 18.15 & 15 & 7.29 \\
    \hline
    \end{tabular}
    \caption{Groups' costs and two possible ways to share them among the travellers, for the example shown in Figure \ref{Img:ExampleIntro}. We mark with a grey background the optimal solution when all costs are accounted for.}
    \label{tab:ExampleIntro}
\end{table}

What would happen if users can choose with whom to travel? Second column of Table \ref{tab:ExampleIntro} shows the total cost per group. Each user does not perceive that number, but only what is allocated to her. Therefore, from the users' perspective, costs per group have to be split somehow among its members, and these individual costs are the relevant feature for the users, so that they will always try to choose a group that minimises their respective costs.

If groups' costs were split according to the first alternative (columns 3, 4 and 5 of Table \ref{tab:ExampleIntro}), there could be no stable solution. On the one hand, the group $\{1,2,3\}$ might never be formed, because users $1$ and $2$ would rather leave the group to travel together. On the other hand, travelling in pairs is in this case always more convenient than travelling alone, but user $3$ prefers to travel with user $2$, user $2$ prefers to travel with user $1$, and user $1$ prefers to travel with user $3$. Note that both the definition of a stable solution and the way in which the costs are split play a role here, which are the two main topics we discuss in this paper. 

The relevance of how the costs are shared is highlighted by the other cost-sharing alternative (last three columns of Table \ref{tab:ExampleIntro}) that results from following one of the cost-splitting protocols proposed in the paper. In this case, the best alternative for users 1 and 2 is to choose the group that also minimises total costs, which leaves user 3 with no other alternative than travelling by herself. That is to say, the users' interests become in this instance perfectly aligned with a system-wide optimisation.

\subsection{Structure of the paper}
The remaining of the paper is organised as follows. Section \ref{scn:relatedworks} reviews the most important previous works. Section \ref{scn:FormalGame} formalises the game and describes the equilibrium notions that we study. Section \ref{scn:pricing} proposes the respective cost-sharing protocols, which are tested numerically in section \ref{scn:results}. Finally, section \ref{scn:conclusions} concludes and proposes some lines for future research.

\section{Related works}\label{scn:relatedworks}

\subsection{Sharing costs in on-demand systems}

The question of how to split the costs within a shared ride has been studied by some previous papers, using the so-called \textit{cooperative game theory}. The aim of such models is forming coalitions, defined as sets of players that decide to cooperate to form the best possible solution; hopefully, there will be a single coalition formed by all the players. In the context of ridepooling (or ridesharing), this implies that the cost of a user within a group does not depend only on her co-travellers, but also on the other members of the coalition, which might be inconvenient as the protocols can be difficult to understand (and thus to accept) by the passengers.

Such an approach is followed by \cite{lu2019fair}, who study the case in which vehicles follow the shortest circuits to serve their passengers (i.e., solving a traveling-salesman-problem, TSP), which might not be the best case when users' costs are part of the decision. They focus on finding the so-called \textit{nucleus} of the cooperative game, meaning that they minimise the maximal dissatisfaction among all the groups that are formed. \cite{levinger2020computing} also study a similar framework as a cooperative game, and focus on how to compute the so-called \textit{Shapley values}, which are known to be fair prices in such cooperative games, but are usually difficult to compute. Their main result is that when only vehicle-kilometers are taken into account and users are sorted a-priori, Shapley values can be calculated in polynomial time. \cite{bistaffa2017cooperative} deals with both the optimal solution and the \textit{kernel} of the cooperative game on a similar ridesharing system.

A different tool that has been widely used is \textit{mechanisms design}, in which users offer prices as in an auction. \cite{kleiner2011mechanism,shen2016online,cheng2014mechanisms,bian2020mechanism} study mechanisms to match that are \textit{incentive compatible}, i.e., in which each user's best strategy is to reveal their true interests: \cite{kleiner2011mechanism} focus on a ridesharing system, \cite{shen2016online} on ridepooling, while \cite{cheng2014mechanisms,bian2020mechanism} on feeder-trunk systems where the on-demand component serves the last mile; in the latter, users do not share the vehicle simultaneously.  Regarding non-shared ridehailing, \cite{asghari2016price} proposes bid mechanisms to assign drivers to riders, while \cite{zhang2015discounted} study a \textit{discounted trade reduction mechanism} scheme to satisfy every agent. 

Other studies use different techniques to set pricing that attain a stable matching among different agents: vehicles and users are matched by \cite{rasulkhani2019route} in a generic many-to-one transport system (i.e., where many passengers can utilise the same route), while \cite{peng2020stable} propose how to define the payments from passengers to drivers for ridesharing. \cite{chen2018price} consider a model in which users can choose among a set of options that offer different prices and pick-up times. Stable assignments between pairs of riders (or between one driver and one rider) are studied by \cite{wang2018stable,zhang2018mobility,chau2020decentralized,yan2020matching}. \cite{furuhata2015online} study cost-sharing mechanisms for ridesharing, that are updated online as new passengers enter the system, guaranteeing that fares can never increase for a passenger. \cite{ke2020pricing} study the emerging market equilibria in both pooled and non-pooled systems, including drivers' and riders' interests. 

In all, the relationship between efficient prices and stable/optimal assignments in on-demand mobility systems has been increasingly studied during last years. However, the pricing strategies have mostly relied either on prices that might depend on users travelling in other vehicles, or on auctions, so that the question about direct prices for a shared trip remains yet to be studied.

\subsection{Game theory and flexible mobility systems}
Flexible systems require deciding how to match vehicles and users, so that conflicts between the interests of different stakeholders usually emerge. Therefore, using game theory is a natural idea that has been followed by several papers in the past to study different issues related to these mobility systems. \cite{hernandez2018game} consider an abstract model of carpooling, in which the users' decisions (or \textit{strategies}) are whether to cooperate or not. If they do not cooperate, they can decide selfishly, but receive a punishment, whereas cooperating entails a reward. They focus on the evolution in time of users' decisions. \cite{schroder2020Anomalous} use game theory to understand the emergence of surge pricing in ridehailing systems, through a model in which drivers can decide when to turn-off their devices (a similar analyses is performed by \cite{castillo2017surge}, but without employing game theory).  \cite{kucharski2020if} study the system-wide impact of users arriving late at their pick-up locations, which is also modelled using game theory, in which users' decide strategically how late to arrive, taken into account the annoyance of both waiting for other passengers and arriving late at the final destination. \cite{jacob2021ride} studies a for-profit ridehailing system in which users strategically decide whether to travel solo, pooling, or not using the system at all.

As we explain in section \ref{scn:FormalGame}, the game we study can be analysed as a game-theory version of the classical combinatorial problem \textit{set cover}. A similar game has been studied by \cite{escoffier2010impact}, although they focus on the so-called \textit{Price of Anarchy} (i.e., which is the worst possible equilibrium). Moreover, in their model the cost of a user within a group can increase if the co-travellers decide differently, which is reasonable for the abstract model they study, but does not seem to be realistic for modelling a mobility system.

In spite of its relevance, (non-cooperative) game theory has not been widely used to study the matching issues emerging in mobility systems in which users share the same vehicle.

\subsection{Contribution}
This paper's contribution is threefold: First, we identify that optimal ridepooling solutions are not always stable for travellers and formalise the problem as a game. Second, we show that the traditional notions of Nash and Strong Equilibria are not suitable for this problem, and propose three intermediate notions of equilibrium. The third and main contribution of this paper is proposing cost-sharing protocols, so that for each of the three equilibrium notions we can reach a \textit{price of stability} (PoS) equal to 1, i.e., that make any optimal matching an equilibrium. The rules imposed by these protocols to split the costs within a shared trip depend only on the characteristics of the trip itself, and not on the other groups and trips in the system. On top of that, we run numerical simulations over a real-life case from Amsterdam, that confirm our results, showing that optimal solutions are in fact stable for users. Each protocol yields a PoS that is close to 1 regardless of the equilibrium notion, whereas that the price of anarchy can be significantly larger.

\section{The game and the equilibria} \label{scn:FormalGame}

Consider a set of passengers $P=\{1,...,n\}$ that need to be transported by a ridepooling system, i.e., need to be matched into groups that will be transported each by a different vehicle. We assume that there are enough vehicles (at least $n$), so the costs of the system depend only on how the requests are pooled.  There is a set of feasible groups denoted $\mathcal{G}=\{G_1,...,G_m\}$, which depend on some exogenous conditions (like the capacity of the vehicles, or declining groups that require too long detours), i.e., we assume that $\mathcal{G}$ is fixed. Each feasible group $G$ is a subset of $P$, and all its subsets are assumed to be feasible as well, i.e. $G \in \mathcal{G}, H \subseteq G \Rightarrow H \in \mathcal{G}$. In particular, each passenger might choose to travel alone, i.e. $\{i\} \in \mathcal{G} \: \forall i \in P$.  

In the following, we are interested in comparing different ways to match the users. A \textit{matching} is a selection of groups $\{H_1,...,H_\eta\} \subseteq \mathcal{G}$, such that $\eta>0$ is any integer and each $i \in P$ belongs to exactly one of these groups $H$. 

\subsection{Preliminaries}
\subsubsection{Defining a single cost function}
Mobility systems induce costs of different nature and to various agents. We are interested here in comparing the total costs of the resulting matching when groups are formed optimally versus when each user tries to minimise her own individual costs. For such a comparison, monetary transfers between agents are not relevant as they do not affect total costs. However, they are indeed relevant from the users' individual point of view when searching for an equilibrium.

Let us be more specific. Consider a group $G \subseteq P$. If a vehicle is instructed to serve that group, the following costs might be taken into account:
\begin{itemize}
    \item Costs faced by the vehicle's operator, such as fuel, maintenance, depreciation and drivers' wages.
    \item Users' costs: they should at least include total travelling time, involving waiting and in-vehicle. Other specific ridepooling-related aspects might be considered as well, such as the system's unreliability (\cite{fielbaum2020unreliability,alonso2020value}) or users' willingness to share the vehicle with strangers (\cite{ho2018potential,lavieri2019modeling,alonso2020determinants,winter2020identifying,konig2020travellers}).
    \item Other societal impacts, such as congestion or emissions.
\end{itemize}

All those costs have to be monetarised so that they can be compared. All of which depend on the route followed by the vehicle. For the purpose of this study, such a route is exogenous, so that the mentioned costs are given. Usual criteria to define the route are minimising total costs, or including some additional rules like first in-first out.

Therefore, a given group $G$ implies a total cost $c(G)$, encompassing all the costs aforementioned:

\begin{equation}
\label{eq:cost}
    c(G) = \sum_{i \in G} C(G,i) + C_O(G) + C_S(G)
\end{equation}

Where $C(G,i)$ are the direct costs faced by user $i$, like her waiting and in-vehicle times, $C_O(G)$ are the costs faced by the operator and $C_S(G)$ are the societal costs. Not every group can be feasibly matched together, i.e. is mutually compatible. A natural assumption is that if $G$ is feasible, then every $H \subseteq G$ is feasible as well, with $c(G) \leq c(H)$ (and $c(\emptyset)=0)$. Differently from other related papers on the topic, we do not assume any type of supermodularity or submodularity, i.e., if $G_1$ and $G_2$ are disjoint groups and $G_1 \cup G_2$ is feasible, it might hold either that  $c(G_1)+c(G_2) \leq c(G_1 \cup G_2)$ or that $c(G_1)+c(G_2) \geq c(G_1 \cup G_2)$, because the sign of this relationship represents how efficient is to serve the whole group of passengers together (which depends, for instance, on how close are their origins and destinations).

To exemplify these ideas, let us consider the illustrative case depicted in Figure \ref{Img:ExampleIntro}. The detailed costs are exhibited in Table \ref{tab:ExampleCosts}, where we assume that operators' costs include a fixed component $c_K$ and operational costs that are proportional to the driving time by a factor $c_T$. Users' costs increase proportional to the waiting time by a factor $c_W$ and to the in-vehicle time by a factor $c_V$, and we assume no extra societal costs ($C_S=0$). Each trip is assumed to begin at the origin of the first user within the respective group, who has zero waiting time. Colours in the second row of Table \ref{tab:ExampleCosts} match the respective users in Figure \ref{Img:ExampleIntro}. The values exhibited in Table \ref{tab:ExampleIntro} in the Introduction are computed by assuming $c_K=7$ and $c_V=c_W=c_T=1$.

\begin{table}[h]
    \centering
    \begin{tabular}{|c|c|c|}
    \hline
     \textbf{Group}    & \textbf{Users' costs} & \textbf{Operators' costs} \\
     \hline
    $\{1\}$     & $\color{blue}8\color{black}c_V$ & $c_K+8c_T$\\
    $\{2\}$     & $\color{red}6\color{black}c_V$ & $c_K+6c_T$ \\
    $\{3\}$     & $\color{brown}2.24\color{black}c_V$ & $c_K+2.24c_T$\\
    $\{1,2\}$     & $(\color{blue}8\color{black}+\color{red}6\color{black})c_V\color{black}+\color{red}2\color{black}c_W $ & $c_K+8c_T$\\
    $\{1,3\}$     & $(\color{blue}8.32\color{black}+\color{brown}2.24\color{black})c_V\color{black}+\color{brown}6.08\color{black}c_W$ & $c_K+8.32c_T$\\
    $\{2,3\}$     &  $(\color{red}6.36\color{black}+\color{brown}2.24\color{black})c_V\color{black}+\color{brown}4.12\color{black}c_W$  & $c_K+6.36c_T$\\
    $\{1,2,3\}$     & $(\color{blue}8.36\color{black}+\color{red}6.36\color{black}+\color{brown}2.24\color{black})c_V+(\color{red}2\color{black}+\color{brown}6.12\color{black})c_W$ & $c_K+8.36c_T$\\
    \hline
    \end{tabular}
    \caption{Users' and operators' costs for each group in the illustrative example from Figure \ref{Img:ExampleIntro}. Colours in the second column match the respective users in Figure \ref{Img:ExampleIntro}.}
    \label{tab:ExampleCosts}
\end{table}

The costs $\left(c(G):G \in \mathcal{G}\right)$ are the relevant ones from a system's standpoint. The optimal matching can be found using the ILP shown in Eq. \eqref{Eqn:Opt}, where $x_G=1$ if and only if group $G$ is selected to be executed, and the constraint ensures that each passengers is transported in exactly one group. The illustrative example is useful to gain some intuition: if we only considered users' costs ($c_K=c_T=0$), then the optimal solution would be that everybody travels alone. In the opposite case, when only accounting for operators' costs ($c_V=c_W=0$), the optimal solution could be either forming a single group with the three users, or transporting the groups $\{1,2\}$ and $\{3\}$, depending on the relative values of $c_K$ and $c_T$ . The solution of Eq. \eqref{Eqn:Opt} represents the case in which all the costs are accounted for.

\begin{align}
\label{Eqn:Opt}
    \min_{x_G \in \{0,1\}} & \sum_{G \in \mathcal{G}} x_G c(G) \\
\nonumber
    \text{s.t. } & \sum_{G: i \in G} x_G = 1 & \forall i \in P 
\end{align}

Apart from the costs mentioned above, some pricing scheme can be decided, i.e., monetary fares $f(G,i) \: \forall i \in G$. The real costs faced by passenger $i$ when sharing with the group $G$ are then

\begin{equation}
    c_i(G)=C(G,i)+f(G,i)
\end{equation}

In the example represented by Table \ref{tab:ExampleCosts}, determining the fares could be naturally represented by splitting somehow the numbers in the third column (operators' costs) among the users of the respective group.

We assume that fares and other users' costs are perfectly interchangeable, given the appropriate monetary equivalencies. That is to say, we assume that a user considers all the elements that are costly for her (at least time and monetary costs) and computes a single figure encompassing all of it. This single figure allows the user to compare her different groups, i.e., user $i$ would always opt for the available group that minimises $c_i(G)$, regardless of its specific decomposition between $C(G,i)$ and $f(G,i)$. 

As the direct costs $C(G,i)$ are assumed to be exogenous, defining a pricing scheme is equivalent to deciding $c_i(G)$ for every $i$ and $G$. In the rest of the paper, we deal only with the two relevant functions affecting the decisions taken by the system and by the users:
\begin{itemize}
    \item $\left(c(G): G \in \mathcal{G}\right)$, which are the relevant amounts from a system's standpoint, as all the costs are encompassed there. The optimality and efficiency of each feasible matching depends solely on these numbers.
    \item $\left(c_i(G): i \in G, G \in \mathcal{G}\right)$, which are the relevant amounts from the passengers' standpoint. When users decide in which group they want to belong, i.e., when searching for equilibria, this is the only figure that matters.
    \end{itemize}
In other words, the direct costs $C(G,i), C_O(G), C_S(G)$, as well as the fares $f(G,i)$ will be implicit in all future definitions and calculations.

The problem faced by an operator that is aiming for an optimal solution that is also stable can be seen as two-level optimisation: First, the operator decides on some pricing scheme, and then the users choose how to match. \textbf{The question that needs to be solved by the system is then how to set prices so that users match in a way that yields low total costs. }

Eq. \eqref{Eqn:Opt} serves as the benchmark to compare how much worse is a resulting (stable) matching compared with the system optimum solution. Therefore, it is useful to see that Eq. \eqref{Eqn:Opt} can be read exactly as the set cover problem, which is known to be NP-Hard, and even further, that it is impossible to find a polynomial algorithm that provides a solution that approximates the optimal solution by a constant factor, unless $P=NP$ (\cite{raz1997sub}). That is to say, our problem can be seen as a version of set cover in which each element of the set might decide which subset is covering it\footnote{In the traditional set cover problem, elements of the universal set might be covered by more than one subset, which we do not allow here when we impose an equality in the constraint in Eq. \eqref{Eqn:Opt}. However, it is straightforward to see that the problems are indeed equivalent, modifying the usual definition of set cover: For each $G \in \mathcal{G}$, for each $H \subseteq G$, if $H \notin \mathcal{G}$, we add $H$ to $\mathcal{G}$ with $c(H)=\min_{G' \in \mathcal{G}: H \subset G'} c(G')$. Doing so, there is always an optimal solution covering each element exactly once.}.

Although Problem \ref{Eqn:Opt} could be hard to solve (because set cover is NP-Hard), real-life ridepooling situations usually generate versions in which standard solvers are able to manage it (\cite{alonso2017demand,kucharski2020exact,fielbaum2020OnDemand}). Therefore, we shall assume in the rest of the paper that the optimal solution can be computed in reasonable time.

The rest of the section is devoted to the second-level of the problem mentioned above, i.e., from the users' point of view, meaning that prices are exogenous and fixed. We aim at understanding users' decisions and characterising the equilibria. We now discuss the formal game, and show that there is no unique notion of equilibrium, as they depend on which level of coordination is assumed users might reach. The problem of defining such individual costs $c_i(G)$ is discussed in section \ref{scn:pricing}. 


\subsubsection{The formal Co-Travellers Game CTG}
We can now formally define the underlying game, denoted as CTG: \textit{Co-Travellers Game}. The players of CTG are the passengers $p_1,...,p_n$, who are the ones taking the decisions. Each passenger $i$ can decide in which group she wants to travel, i.e., her possible strategies are the sets $G$ in $\mathcal{G}_i$, defined as $\mathcal{G}_i=\{G \in \mathcal{G}: i \in G\}$.

Following the usual notation in game theory, we denote $G_i$ the strategy (i.e., the group) chosen by $i$, and $G_{-i}=(G_j)_{j \neq i}$ the profile of strategies chosen by the other players. When $i$ chooses $G_i$, there are two possibilities: either all the other players in $G_i$ choose $G_i$ as well, or at least one player does not. We assume that no one can be forced to join a group, so $G$ is executed only in the first case, yielding a disutility function $U$:

\begin{equation} \label{Eqn:Disutility}
 U_i(G_i,G_{-i})
 \begin{cases}
 c_i(G_i) & \text{ if } G_j=G_i \forall j \in G_i \\
 +\infty & \text{ otherwise }
 \end{cases}   
\end{equation}

The $+\infty$ disutility in Eq. \eqref{Eqn:Disutility} represents the case in which the group is not executed, so $i$ cannot travel. We say that a profile of strategies $(G_i)_{i \in P}$ is \textit{coordinated} if this situation does not happen in the whole system, i.e., if nobody chooses a group that is not being executed. Note that a coordinate profile of strategies is a matching (i.e., it satisfies the constraint in Eq. \ref{Eqn:Opt}).

Although non-coordinated situations are feasible, they shall not occur because any passenger would opt for travelling alone rather than choosing a group that is not being executed. This is formalised in the following two Lemmas:

\begin{lemma}
Any Nash Equilibrium of the CTG is formed only by pure strategies.
\end{lemma}

\begin{proof}
Suppose player $i$ chooses a mixed strategy, with $Prob(G_i=G)=p$ for some $G \in \mathcal{G}_i, p \in (0,1)$, $G \neq \{i\}$. Let $j \in G$ be a different player. Then 
\begin{equation*}
  E(U_j(G)) \geq (1-p) \cdot U_j(G|G_i \neq G) = (1-p) \cdot +\infty = + \infty  
\end{equation*}
That is to say, if $j$ chooses $G$, then it faces an infinity expected disutility, because there is a chance that the group is not executed. Hence, $G$ cannot be chosen by $j$ with a non-zero probability in an equilibrium, because travelling alone yields a lower disutility.
The same argument can now be applied to $i$. Given that we already know that $j$ is assigning zero probability to $G$:

\begin{equation*}
    U_i(G_i, G_{-i})  = + \infty
\end{equation*}

Hence, $i$ can improve her situation by assigning the mentioned probability $p$ to $\{i\}$ instead of $G$, which proves that this situation is not an equilibrium.
\end{proof}

Consequently, in the remaining of the paper we assume that players never choose mixed strategies.

\begin{lemma}
If a profile of strategies is a Nash Equilibrium, then it is coordinated.
\end{lemma}

\begin{proof}
Let $(G_i)_{i \in P}$ be a non-coordinated profile of strategies. Let $i$ be a player choosing a group that is not being executed. Then $i$ would be better-off by choosing $\{i\}$ instead, implying that $(G_i)_{i \in P}$ is not an equilibrium.
\end{proof}

To simplify the notation, and because our analyses deal with users reaching equilibria, we drop the disutility $U_i$ and will utilise the costs $c_i$ directly, because they always coincide. As we now discuss (section \ref{Scn:TE}), the traditional notions of equilibrium are not the most suitable ones for this game, which is why we will propose alternative notions afterwards (section \ref{SCN:RE}). We will create abbreviations, that will mark with an initial ``T'' the traditional notions of equilibrium, and with an ``R'' (from ``Ridepooling'') the alternative ones.

An equilibrium notion is characterised by the level of coordination that is admitted among the users. For instance, the traditional Nash Equilibrium assumes that users cannot coordinate at all, so each user can only act unilaterally, whereas the traditional Strong Equilibrium assumes full coordination. The alternative notions we will propose lay between those two. For any equilibrium notion, we are mostly interested in the PoS, i.e., in the ``best'' possible equilibrium, meaning the equilibrium that yields lowest total costs, which is interpreted as the matching that the system should propose to the users. When we run numerical experiments, we will also look at the \textit{price of anarchy} (PoA), which is the ``worst'' possible equilibrium, i.e., the one that yields the highest total costs, which represents the worst case scenario when users coordinate freely without any recommendation from the system.


\subsection{Traditional notions of equilibrium} \label{Scn:TE}
\subsubsection{Traditional Nash Equilibrium (TNE)}.
The most common notion in game theory is the Nash Equilibrium:
\begin{definition}
A matching $(G_i)_{i \in P}$ is a TNE if no player can unilaterally improve her situation, i.e.

\begin{equation} \label{Eqn:TNE}
\forall i \in P, \forall G \in \mathcal{G}_i, c_i(G_i,G_{-i}) \leq c_i(G,G_{-i})   
\end{equation}
\end{definition}

We know that a TNE exists, because every finite game has at least one mixed or pure Nash Equilibrium, and this game does not admit mixed equilibria. Moreover, consider a TNE and a player $i$: Eq. \eqref{Eqn:TNE} is trivially fulfilled by any $G$ that is neither $G_i$ nor $\{i\}$, because if $i$ would switch to such a $G$, her group would not be executed and her cost would be $+\infty$. Therefore, Eq. \eqref{Eqn:TNE} can be replaced by 

\begin{equation} \label{Eqn:TNE2}
\forall i \in P, c_i(G_i,G_{-i}) \leq c_i(\{i\},G_{-i})   
\end{equation}

We define now the set of groups in which Eq. \eqref{Eqn:TNE2} holds:

\begin{equation}
    \mathcal{G}_{TNE}=\{G \in \mathcal{G}: \forall i \in G, c_i(G) \leq c_i(\{i\}) \}
\end{equation}

Any solution to problem \ref{Eqn:Opt} that only selects groups from $\mathcal{G}_{TNE}$ would be a TNE. Hence, we can find the PoS and the PoA by solving, respectively, problems \ref{Eqn:BPoATNE} and \ref{Eqn:WPoATNE}:

\begin{align}
\label{Eqn:BPoATNE}
    \min_{x_G \in \{0,1\}} & \sum_{G \in \mathcal{G}_{TNE}} x_G c(G) \\
\nonumber
    \text{s.t. } & \sum_{G: i \in G} x_G = 1 & \forall i \in P 
\end{align}

\begin{align}
\label{Eqn:WPoATNE}
    \max_{x_G \in \{0,1\}} & \sum_{G \in \mathcal{G}_{TNE}} x_G c(G) \\
\nonumber
    \text{s.t. } & \sum_{G: i \in G} x_G = 1 & \forall i \in P 
\end{align}

Problem \ref{Eqn:BPoATNE} results from constraining problem \ref{Eqn:Opt} to consider only groups in $\mathcal{G}_{TNE}$, i.e. finding the best solution among this subset of groups, while problem \ref{Eqn:WPoATNE} switches from a minimisation to a maximisation problem, to find the worst solution.



Note that $\{i\} \in \mathcal{G}_{TNE} \forall i \in P$. This is troublesome, because it entails that a trivial TNE is that everybody travels alone, which might not be a reasonable representation of reality. If we think of the ridepooling service as provided by an app, a TNE would imply that a user that is suggested to join a group can only choose to accept such a group or to dismiss it and travel alone. This occurs because TNE is a weak notion, that forbids any type of coordination among users, which precludes accounting for the complexity in the equilibrium analysis that is introduced by the change of sharing. All of this justifies analysing other types of equilibria.

\subsubsection{Traditional Strong Equilibrium (TSE)}\label{scn:TSE}
An alternative usual notion in game theory, that represents the other extreme situation, is the Strong Equilibrium, in which each subset of players can coordinate:

\begin{definition}
A matching $(G_i)_{i \in P}$ is a TSE if no group of players can jointly improve their situation. As the utility of a player only depends on her co-travellers, this happens if and only if
\end{definition}

\begin{equation} \label{Eqn:TSE}
    \forall G \in \mathcal{G}, \text{ either } \forall i \in G,  G_i=G, \text{ or } \exists i \in G \text{ such that } c_i(G_i)\leq c_i(G)
\end{equation}

Eq. \eqref{Eqn:TSE} ensures that if a group $G$ is not taking place, then it cannot happen that all the users from $G$ would want to abandon their current groups to form $G$. This is a quite restrictive notion of an equilibrium because it permits any type of coordination between all the users in the game. Therefore, it does not come as a surprise that there are cases in which there is no such an equilibrium\footnote{When the system can decide on the prices, there is a trivial way to induce a TSE, by sharing the total costs of a group uniformly among the users. After that, a greedy algorithm that repeatedly picks the group of non-selected users with the lowest average cost will output a TSE. It is noteworthy that this mimics exactly the well-known polynomial greedy algorithm for set cover, implying that its cost is no larger than $\log(n+1)$ times the optimal (\cite{chvatal1979greedy,raz1997sub}). However, as discussed in the Introduction, such a way to split the costs would probably not be accepted by users travelling short distances.}, as shown in the following example:

\begin{example} \label{Example:TSE}
Consider an instance of CTG with three players $A, B$ and $C$. For all of them, travelling alone is the worst option, and it is infeasible that they all share a vehicle. Regarding pairwise trips, $A$ prefers to travel with $B$ over travelling with $C$, $B$ prefers $C$ over $A$, and $C$ prefers $A$ over $B$. Such an instance presents no TSE: if they all travel alone, any pair would prefer to join, while if two players are traveling together, say $A$ and $B$, then $B$ would coordinate with $C$ to create the group $\{B,C\}$, and analogous situations would occur for any other pair.
\end{example}

Not only an equilibrium might not exist. Determining whether that is the case is computationally intractable, unless all the groups are formed by one or two users. 
\begin{theorem} \label{Thm:TSENPComplete}
a) Given an instance of CTG, determining whether a TSE exists is NP-Complete. This is true even if groups' sizes cannot be larger than 3. 
b) On the other hand, if no group has more than 2 players, it is possible to determine the existence of a TSE in polynomial time.
\end{theorem}

\begin{proof}
We first argue that this problem is NP, by noting that it is possible to verify in polynomial time if a given profile of strategies is a TSE. Indeed, we need to determine if Eq. \eqref{Eqn:TSE} is fulfilled, which is done by checking once for each $G \in \mathcal{G}$ if the respective conditions described by Eq. \eqref{Eqn:TSE} are met.

a) To prove that the problem in which groups can have 3 or fewer players is NP-Complete, we use a reduction from \textit{3-Exact-Cover}. The $k$-Exact-Cover is defined by a universe set $S$ whose size is a multiple of $k$, and a collection of subsets $\mathcal{B}=b_1,\ldots,b_Q$, with $|b_i|=k \: \forall i=1,\ldots,Q$. The question is if one can pick some of those subsets, such that they are all disjoint and cover the set $S$. The $k$-Exact-Cover is known to be NP-Complete for any $k\geq 3$ (\cite{johnson1979computers}).

Given an instance of 3-Exact-Cover, we build an instance of CTG as follows:
\begin{itemize}
    \item The set of players is $S^*=S \cup \{\alpha\}$, with $\alpha$ an additional artificial player.
    \item Each subset of $S^*$ whose size is not larger than 3 is a feasible group.
\end{itemize}

In order to define the costs, please note first that, for this proof, we might drop the need that if $H \subseteq G$, then $c(H) \leq c(G)$. In fact, if our cost scheme does not meet this condition, we can add a large quantity $D$ to $c_i(g)$ for each $i$ and $g$. By doing so, the equilibria analysis remains unchanged because all individual costs are raised by the same amount, and the cost of each group $g$ will have a part that is $D \cdot |g|$, such that if $D$ is large enough then this proportional part will outweigh any other differences between groups of different size. Moreover, we do not need to define the costs numerically, as it suffices to show for each user how does she rank the different groups that she belongs to:
\begin{itemize}
    \item Groups of size $3$ corresponding to some $b_i$ (i.e., that comes from the instance of the original $3$-cover problem) are equally ranked as the preferred groups for all their members.
    \item Groups of size $1$ are the less convenient ones for everybody.
    
\end{itemize}
To rank the remaining groups, i.e., groups of size 3 that do not correspond to any $b_i$, and pairs, we first introduce a notion of individual preferences for each user: We say that player $i$ prefers player $j_1$ over $j_2$ if and only if 

\begin{equation}\label{Eq:Circle}
    j_1 - i \text{ (mod $n$) } < j_2 - i \text{ (mod $n$) }
\end{equation}

And we denote this situation by $j_1 \prec_i j_2$. Eq. \eqref{Eq:Circle} can be interpreted as if all users were displayed in a circle, such that $i$ always prefers the first player that appears to her right. Therefore, $i+1$ (mod $n$) is the most preferred and $i-1$ (mod $n$) is the least preferred player for player $i$.

For player $i$, the remaining groups are sorted \textit{lexicographically} according to the relationship $\prec_i$. That is to say, consider two different groups $J=(i,j_1,j_2)$ and $K=(i,k_1,k_2)$, such that $j_1 \prec_i j_2$ and $k_1 \prec_i k_2$. Then, $i$ prefers $J$ over $K$ if and only if
\begin{equation}
    j_1 \prec_i k_1 \text{ or } \left( j_1=k_1 \text{ and } j_2 \prec_i k_2 \right)
\end{equation}

Intuitively, for player $i$ the most relevant aspect to evaluate the group that she belongs to is the closest of her co-players. If two groups present the same closest co-player, then the remaining co-player is relevant, following the same order. Such a first criterion is also the decisive one when there is a group of two users involved, i.e., if the same group $J$ is compared with the group $K'=(i,k)$, $i$ prefers $J$ if and only if

\begin{equation}
    j_1 \prec_i k \text{ or } j_1=k
\end{equation}

Finally, if $i$ has to choose between $(i,j)$ and $(i,k)$, she will chose $(i,j)$ if and only if $j \prec_i k$. The details that prove that this reduction works can be found in the appendix.

b) To prove that the problem of determining whether a TSE exists is polynomial when all the groups have 1 or 2 users, we show that such a problem is equivalent to the well-known \textit{stable-roommate problem}, which can be solved in polynomial time (\cite{gusfield1988structure}). The details proving such an equivalence can be found in the appendix.
\end{proof}

It is noteworthy that determining whether a general game has a TSE can be harder than being NP-Complete (\cite{gottlob2005pure}), because verifying that a given profile is a TSE is not as simple as here. The main difference is that in CTG we discard all the non-coordinated profiles of strategies. 

\color{black}





The reason why such a notion of equilibrium is not appropriate for this game is that it cannot take place in ridepooling nor ridesharing systems. If hundreds or thousands of users are involved, it is impossible that all of them are able to coordinate just to perform a trip. That is, involving every possible coordination to define stability is too strong.

\subsection{Alternative notions of equilibrium} \label{SCN:RE}

The two most traditional notions of equilibrium are not suitable for this model. TNE is too weak, and TSE is too strong. Therefore, we now propose three alternative definitions for equilibrium in the ridepooling context. All of them lie in-between the two previous ones, permitting some level of joint decisions among the users.

\subsubsection{Ridepooling Hermetic Equilibrium (RHE)}
The first alternative notion of equilibrium recognises that a user could coordinate with those other users that are already related to her: the ones with whom she is currently sharing.

Consider $G \in \mathcal{G}$ and $H \subsetneq G$. We say that $H$ wants to leave $G$ if $\forall i \in H, c_i(H) < c_i(G)$, and $G$ is \textit{hermetic} if no such an $H$ exists. Formally, $G$ is hermetic if and only if:
\begin{equation}
 \forall H \subsetneq G, \exists i \in H \text{ such that }c_i(G)\leq c_i(H)   
\end{equation} 
This idea permits coordination between users, but only when they are in the same group. Verifying if a group is hermetic can be done by testing for all its proper subsets if they want to leave or not.
    
\begin{definition}
A matching $(G_i)_{i \in P}$ is a \textit{hermetic ridepooling equilibrium} if every $G_i$ is hermetic.
\end{definition}

It is evident that having each traveller by herself is a RHE. However, it is not necessarily unique. To find the PoS and PoA, we constraint Problem \ref{Eqn:Opt} in a similar way as for TNE, but with tighter restrictions. Let us define $\mathcal{G}_H$ as the set of hermetic groups, which is found just by pruning the non-hermetic groups. The PoS and PoA are found solving problems \ref{Eqn:BPoAHRE}-\ref{Eqn:WPoAHRE}, that are analogous to \ref{Eqn:BPoATNE}-\ref{Eqn:WPoATNE},  constraining the set of groups to $\mathcal{G}_H$.

\begin{align}
\label{Eqn:BPoAHRE}
    \min_{x_G \in \{0,1\}} & \sum_{G \in \mathcal{G}_{H}} x_G c(G) \\
\nonumber
    \text{s.t. } & \sum_{G: i \in G} x_G = 1 & \forall i \in P 
\end{align}

\begin{align}
\label{Eqn:WPoAHRE}
    \max_{x_G \in \{0,1\}} & \sum_{G \in \mathcal{G}_{H}} x_G c(G) \\
\nonumber
    \text{s.t. } & \sum_{G: i \in G} x_G = 1 & \forall i \in P 
\end{align}

\subsubsection{Ridepooling Unmergeable Equilibrium (RUE)}
We now propose an alternative notion of equilibrium that admits some trivial Pareto improvements: merging groups when it is better for everybody. Allowing for these simple movements prevents to have as an equilibrium the matching in which everybody travels alone.

We say that two disjoint groups $G_1, G_2$ are \textit{mergeable} if $G_1 \cup G_2 \in \mathcal{G}$ and 
\begin{equation}\label{Eq:Mergeable}
\forall i \in G_1 \cup G_2, \: c_i(G_1 \cup G_2) \leq c_i(G_i), \text{ and } \exists i \in G_1 \cup G_2 \text{ such that } c_i(G_1 \cup G_2) < c_i(G_i)    
\end{equation}
This represents a way in which users from two different vehicles can coordinate, but only if everyone benefits from doing so. As this is a Pareto improvement (i.e. some players get better-off while nobody loses), these mergers might be suggested for an external controller.

\begin{definition}
A matching $(G_i)_{i \in P}$ is a \textit{ridepooling unmergeable equilibrium} if it is a TNE and no two different groups $G_i, G_j$ are mergeable.
\end{definition}

Such an equilibrium always exists, as shown in Theorem \ref{lemma:URE}.

\begin{theorem} \label{lemma:URE}
Any instance of CTG admits a RUE.
\end{theorem}
\begin{proof}
In the appendix.
\end{proof}

As this equilibrium condition does not depend on single groups but rather on the chance of merging them, we cannot just restrict $\mathcal{G}$ to find PoS and PoA. Instead we first identify all the pairs of groups that are mergeable, and then use an extra constraint to prevent that they both occur simultaneously. In Algorithm \ref{Alg:Mergeable} we build the set $\mathcal{M}$ of mergeable groups, i.e., each element of $\mathcal{M}$ is a duple formed by two groups that are mergeable.

\begin{algorithm}
\begin{algorithmic}
\caption{Construction of $\mathcal{M}$.}\label{Alg:Mergeable}
\STATE{$\mathcal{M}=\emptyset$};
\FORALL{$G_1,G_2 \in \mathcal{G_{TNE}}$}
\IF{$G_1\cap G_2 = \emptyset$ and $G_1,G_2$ are mergeable}
\STATE{$\mathcal{M} \leftarrow \mathcal{M} \cup (G_1,G_2)$};
\ENDIF
\ENDFOR
\STATE{Output $\mathcal{M}$};
\end{algorithmic}
\end{algorithm}

If $(G_1,G_2) \in \mathcal{M}$, both occurring simultaneously would mean that the matching is not a RUE. Therefore, we prevent that from happening in the ILP by adding extra constraints ensuring that the respective binary variables cannot both take a value of $1$. This is shown in the inequalities represented by the last lines in Eqs. \ref{Eqn:BPoAURE}-\ref{Eqn:WPoAURE}, that find the PoS and PoA for this type of equilibrium.

\begin{align}
\label{Eqn:BPoAURE}
    \min_{x_G \in \{0,1\}} & \sum_{G \in \mathcal{G_{TNE}}} x_G c(G) \\
\nonumber
    \text{s.t. } & \sum_{G: i \in G} x_G = 1 & \forall i \in P  \\
\nonumber
    \: & x_{G_1}+x_{G_2} \leq 1 & \forall (G_1,G_2) \in \mathcal{M} 
\end{align}

\begin{align}
\label{Eqn:WPoAURE}
    \max_{x_G \in \{0,1\}} & \sum_{G \in \mathcal{G_{TNE}}} x_G c(G) \\
\nonumber
    \text{s.t. } & \sum_{G: i \in G} x_G = 1 & \forall i \in P \\
\nonumber
    \: & x_{G_1}+x_{G_2} \leq 1 & \forall (G_1,G_2) \in \mathcal{M} 
\end{align}

\subsubsection{Ridepooling Semi-individual Equilibrium(RSIE)}
One of the reasons why TNE recognises everyone travelling alone as an equilibrium, is that individual (unilateral) movements are very restrictive in the formal definition of CTG. A natural extension for the idea of individual movements is that a single player switches group, with the players from the receiving group having the chance of accepting her if everybody improves or stays the same. With this idea in mind, we can provide a formal definition of an equilibrium that is not fully, yet nearly, individual (because of the users that have to accept the moving player):

Let $G_1, G_2 \in \mathcal{G}$ and disjoint. We say that $G_1,G_2$ are \textit{individually unstable} if $\exists i \in G_1$ such that $G_2 \cup \{i\} \in \mathcal{G}$ and the two following conditions are fulfilled:

\begin{enumerate}
    \item $c_i(G_2 \cup \{i\})< c_i(G_1)$
    \item $\forall j \in G_2, c_j(G_2 \cup \{i\}) \leq c_j(G_2)$
\end{enumerate}

The first condition entails that $i$ wants to move from $G_1$ to $G_2$, while the second condition ensures that everyone in $G_2$ is willing to accept $i$. This equilibrium can be seen as the one expected when the service is offered through an app, such that each user can take only two actions: i) leaving her current group (as for the TNE), and ii) joining a new group; the latter shall take place only if all users of the mentioned new group accept her.

\begin{definition}
A matching $(G_i)_{i \in P}$ is a ridepooling semi-individual equilibrium if no $G_1, G_2$ are individually unstable.
\end{definition}

The same counter-example \ref{Example:TSE} that revealed that there are instances of CTG with no TSE  is also a scheme with no RSIE. 

\begin{theorem}
Given an instance of CTG, determining whether a RSIE exists is NP-Complete. This is true even if groups' sizes cannot be larger than 3.
\end{theorem}
\begin{proof}
The same arguments used to prove Theorem \ref{Thm:TSENPComplete} a) are valid for this theorem, as it involved only the kind of changes admitted by this equilibrium notion.
\end{proof}


Similar to RUE, finding the PoS and PoA under this notion of equilibrium requires adding some constraints to the ILP. These constraints might lead to an empty set, if no RSIE exists. We first need to identify the pairs of sets that are individually unstable, which are kept in the set $\mathcal{S}$ built in Algorithm \ref{Alg:Unstable}. Here we assume explicitly that $\emptyset \in \mathcal{G}$, entailing that the RSIE is also a TNE, when taking $G_2=\emptyset$ in Algorithm \ref{Alg:Unstable}.

\begin{algorithm}[H]
\begin{algorithmic}
\caption{Construction of $\mathcal{S}$.}\label{Alg:Unstable}
\STATE{$\mathcal{S}=\emptyset$};
\FORALL{$G_1,G_2 \in \mathcal{G}$}
\IF{$G_1\cap G_2 = \emptyset$ and $G_1,G_2$ are individually unstable}
\STATE{$\mathcal{S} \leftarrow \mathcal{S} \cup (G_1,G_2)$};
\ENDIF
\ENDFOR
\STATE{Output $\mathcal{S}$};
\end{algorithmic}
\end{algorithm}

Once the set $\mathcal{S}$ has been computed, the PoS and PoA for this type of equilibrium are obtained solving problems \ref{Eqn:BPoASIRE}-\ref{Eqn:WPoASIRE}.

\begin{align}
\label{Eqn:BPoASIRE}
    \min_{x_G \in \{0,1\}} & \sum_{G \in \mathcal{G}} x_G c(G) \\
\nonumber
    \text{s.t. } & \sum_{G: i \in G} x_G = 1 & \forall i \in P  \\
\nonumber
    \: & x_{G_1}+x_{G_2} \leq 1 & \forall (G_1,G_2) \in \mathcal{S} 
\end{align}

\begin{align}
\label{Eqn:WPoASIRE}
    \max_{x_G \in \{0,1\}} & \sum_{G \in \mathcal{G}} x_G c(G) \\
\nonumber
    \text{s.t. } & \sum_{G: i \in G} x_G = 1 & \forall i \in P \\
\nonumber
    \: & x_{G_1}+x_{G_2} \leq 1 & \forall (G_1,G_2) \in \mathcal{S} 
\end{align}

\subsection{Synthesis and positioning of the different notions of equilibria}

\begin{figure}[H]
    \centering
    \includegraphics[width=55mm]{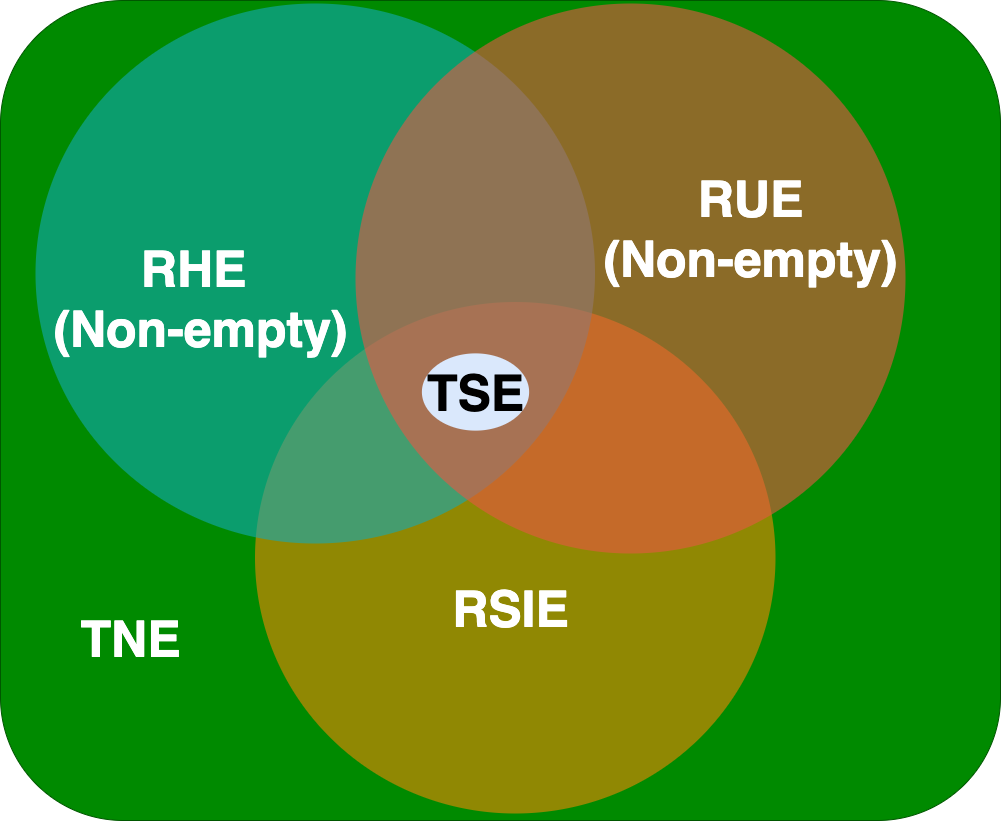}
    \caption{Venn diagram of the five notions of equilibria. TNE, RHE and RUE are never empty, whereas RSIE and TSE might be empty.}
    \label{fig:Equilibria}
\end{figure}

We have shown that TNE is too weak as an equilibrium notion (i.e., it is too easy to fulfill) and that TSE is too strong. We have therefore proposed three extra definitions for equilibria that lie in-between, depicted in Figure \ref{fig:Equilibria} together with TNE and TSE. These three alternative equilibria represent different types of coordination among users: RHE admits full coordination between users within the same group, RUE admits a simple type of Pareto improvements, and RSIE accepts individuals moving from one group to another, if the users from the latter are willing to accept her.


We have shown that there is always at least one RHE and one RUE, entailing that the set of TNE cannot be empty. The TSE and RSIE sets, on the other hand, might be empty and it is an NP-Complete problem to determine if that is the case.

It is noteworthy that the intersections among the groups create new types of equilibria, in which two (or three) of the equilibria conditions are held simultaneously. For the sake of simplicity, we are not studying those intersections explicitly. Note that any of those intersection requiring RSIE might be empty. Moreover, there might also be no matching that satisfies both RHE and RUE: an example (with five players) for this is provided in the Appendix; however, one of the cost-sharing protocols we study in following section does make optimal solutions both a RHE and RUE.

These ideas can be applied to study which equilibria could emerge in the illustrative examples (from section \ref{scn:example}), if the groups' costs are split according to the first alternative in Table \ref{tab:ExampleIntro}:
\begin{itemize}
    \item The TNE notion does not discard any group, i.e., any matching would be feasible, which reinforces that this notion is not appropriate.
    \item The RHE notion would discard the group $\{1,2,3\}$. In particular, the optimal solution is a RHE.
    \item The RUE notion discards any pair of individuals travelling alone (because they could merge), and nothing else. In particular, the optimal solution is a RUE.
    \item The RSIE notion discards the same pairs, but also any combination of a pair and the remaining user by herself. The optimal solution is no longer a RSIE, but everybody travelling in a single group is.
    \item There is no TSE.
\end{itemize}

\section{Cost-sharing protocols to induce good equilibria} \label{scn:pricing}
So far we have assumed that $c_i(G)$ is exogenous to the problem. However, from a policy point of view, one is interested not only in \textit{predicting} which equilibria might emerge, but also in \textit{steering towards} the best possible equilibria. Therefore, the controller of the system can determine the fares for each user (recall that this is equivalent to determining $c_i(G)$ for every $i$ and $G$), in order to induce efficient outcomes. 

From now on, we assume that only $c(G)$ (the total cost of the group) is given for each $G$, which results, for instance, after optimising the route that the vehicle must follow to serve the whole group. We aim to find a so-called \textit{cost-sharing protocol}, i.e., to define how to split the costs amongst the players in $G$ by defining the respective individual costs $c_i(G)$. Following \cite{christodoulou2020resource}, we will say that such protocols are:
\begin{itemize}
    \item \textbf{Budget balanced} if they cover exactly the cost of the group, i.e.
    \begin{equation}
        \forall G \in \mathcal{G}, \sum_i c_i(G) = c(G)
    \end{equation}
    \item \textbf{Overcharging} if they are equal or in excess of the cost of the group, i.e.
     \begin{equation}
        \forall G \in \mathcal{G}, \sum_i c_i(G) \geq c(G)
    \end{equation}

\end{itemize}

    Note that here we are not dealing with the profit of the system. The idea of an overcharging protocol is that the set of individual costs can effectively push users to decide on groups that are efficient from a global point of view.

As we aim at proposing pricing schemes that are understandable by the users, we will consider only \textit{oblivious} protocols\footnote{As discussed in Section \ref{scn:relatedworks}, using oblivious protocols precludes us from using a cooperative game theory approach.} (as defined by \cite{christodoulou2017cost}), meaning that $c_i(G)$ depends only on $i$ and $G$, and not on the other feasible and selected groups in $\mathcal{G}$. To be precise, an oblivious protocol in this context defines $c_i(G)$ as a function of the vector $\left(c(H)\right)_{H \subseteq G}$. Moreover, for the different protocols we propose, all such functions can be computed in polynomial time.



In this section, for each of proposed  equilibria (RUE, RHE or RSIE) we propose a cost-sharing protocol that make any optimal matching an equlibrium (i.e., PoS$=1$). Before introducing the different protocols, it is useful to show and prove the following intuitive Lemma.


\begin{lemma} \label{Lemma:BudgetBalancedRUE}
Let $\{c_i(G): i \in P, G \in \mathcal{G}\}$ be a budget-balanced cost-sharing protocol. Then any optimal matching does not contain pairs of mergeable groups.
\end{lemma}
\begin{proof}
Consider $G_1$ and $G_2$ part of an optimal matching. If $G_1$ and $G_2$ were mergeable (Eq. \ref{Eq:Mergeable}), then the total cost of $G_1 \cup G_2$ would be strictly lower than $c(G_1)+c(G_2)$, which contradicts the optimality of the matching.
\end{proof}

We now proceed to introduce the different protocols.

\subsection{Externality-based protocol}
There is a vast literature that studies how incorporating externalities into pricing might induce optimal equilibria. The so-called Pigouvian taxes (\cite{cremer1998externalities}) for markets regulation, or the VCG-mechanisms for auction designs (\cite{nisan2007computationally}) are some of the most relevant examples. In CTG, however, doing so presents some difficulties:
\begin{itemize}
    \item The externalities induced to the whole system depend on the complete set of passengers and groups, which would require violating the principle of proposing only oblivious protocols.
    \item We have defined several notions of equilibria, so making an optimal matching a TNE, instead of a stronger notion, is not good enough.
\end{itemize}
Notwithstanding the above, we now show how to define costs according to some of the externalities (namely those induced to the members of the same group), using an oblivious protocol and making any optimal matching a RSIE. However, such a protocol will not be budget balanced.

\begin{definition}
Let $G$ be a group. The \textbf{externality-based} protocol charges to each $i$ in $G$ the extra cost that its inclusion induces to the group, i.e.

\begin{equation} \label{Eq:Externality}
c_i(G)=c(G) - c(G \setminus \{i\})    
\end{equation}
\end{definition}

\begin{theorem} \label{thm:Externalities}
If the externality-based protocol is used, then every optimal matching is a RSIE.
\end{theorem}

\begin{proof}
Let $(G_i)_{i \in P}$ be an optimal set of groups, and take player $i$ and group $G_j \neq G_i$ such that $G_j \cup \{i\} \in \mathcal{G}$. We will prove that $i$ does not want to move to $G_j$. If $i$ moves to $G_j$, her new costs would be:

\begin{equation} \label{Eq:Aux1}
    c_i(G_j \cup \{i\}) = c(G_j \cup \{i\}) - c(G_j)
\end{equation}

The difference in costs for $i$ is found by comparing Eqs. \ref{Eq:Externality} and \ref{Eq:Aux1}:

\begin{equation}
   c(G_j \cup \{i\}) - c(G_j) - \left[ c(G_i) - c(G_i \setminus \{i\}) \right] = c(G_j \cup \{i\}) + c(G_i \setminus \{i\}) - \left[c(G_j) + c(G_i) \right]
\end{equation}

The last expression corresponds exactly to the total changes in the global costs of the system, which are positive because $(G_i)_{i \in P}$ is optimal. Hence, player $i$ would face positive extra costs if moving to any other group. 

\end{proof}

Clearly, this protocol is not budget-balanced. In the presence of economies of scale, represented by sub-modular functions (i.e., in which adding an extra passenger is cheaper when the group is larger), the marginal individual costs $c_i(G)$ might become very low for large groups $G$, thus being far from sufficient to cover the real group costs. Note that this low individual costs represent effectively that larger groups should be prioritised when there are economies of scale. Thus, the system would require subsidisation, which is the usual case whenever there are scale economies. If it is not possible to have subsidies, the problem is partially solved with Corollary \ref{cor:Overcharging}, that shows that the externality-based split protocol can be adapted to be overcharging. However, it is no longer oblivious, but \emph{resource-aware}, meaning that the costs within a group $G$ might depend on all the possible costs $c(H)$ for $H \in \mathcal{G}$ (regardless of the actual groups that are chosen by passengers not in $G$). In practice, the only not-oblivious information that is used to make the protocol budget-balanced is $\max_{H \in \mathcal{G}}c(H)$.

\begin{corollary} \label{cor:Overcharging}
There is an overcharging resource-aware protocol, that makes every optimal matching a RSIE. The only not-oblivious information required is $\max_{H \in \mathcal{G}}c(H)$.
\end{corollary}
\begin{proof}
Let $D \geq \max_{H \in \mathcal{G}}c(H)$. Define $c_i(G)$ as if it was externality-based, but adding $D$, i.e.

\begin{equation}
    c_i(G)=c(G) - c(G \setminus \{i\}) + D
\end{equation}

As $D$ is large enough, then individual costs are enough to cover the real ones. And because it is a fixed cost (every player $i$ has to pay it regardless of the group), the analyses regarding equilibria from Theorem \ref{thm:Externalities} are not affected.

\end{proof}

\subsection{Residual-based protocol}
We now study an idea that might look more natural, in which the costs faced by a user are directly related to the cost of the respective private trip. For this, we define the \textit{residual cost} of a group $\Delta c (G)$ as the difference between the individual and the group costs:

\begin{equation} \label{Eq:Residual}
    \Delta c (G) = c(G) - \sum_{i \in G} c_i(\{i\})
\end{equation}

Such residual costs can be positive or negative, depending on whether it is efficient to group those users together. In general, one might expect that they are negative for the groups that do take place, as otherwise it is not reasonable to form that group. Nevertheless, we do not assume that in what follows.

\begin{definition}
Let $G$ be a group. A \textbf{residual-based protocol} consists of sharing only the residual costs, i.e., there exists \textit{residual prices} $p(i,G) \: \forall i \in G$ such that $c_i(G)=c_i(\{i\})+p(i,G)$. All residual prices $p(i,G)$ have the same sign and must fulfil:

\begin{equation} \label{Eq:piG}
    \sum_{i \in G} p(i,G) = \Delta c(G)
\end{equation}
\end{definition}

Note that requiring every $p(i,G)$ to have the same sign and to fulfil Eq. \eqref{Eq:piG} implies that

\begin{equation} \label{Eq:Sign}
  \forall i \in G,   sign(p(i,G)) = sign(\Delta c(G))
\end{equation}

There might be several different residual split protocols, depending on the definition of $p(i,G)$. Some simple ideas might include making $p(i,G)$ proportional to $c_i(\{i\})$, or just a uniform division. In any of which, users first pay a cost that depends solely on their particular travel characteristics, and only the residual component depends on the rest of the group. These protocols are obviously budget-balanced, and we now show that any optimal matching is a RUE when they are applied.

\begin{theorem}
Assume that any residual split protocol is applied. Then, every optimal matching is a RUE.
\end{theorem}

\begin{proof}
Let $(G_i)_{i \in P}$ be an optimal matching. We need to prove that it is a TNE and that no two different groups are mergeable.

To see that it is indeed a TNE, consider some $G_i$ with more than one passenger (groups of size 1 do not need to be analysed). Because this is an optimal matching, the cost of $c(G_i)$ cannot be larger than all of its members travelling alone, i.e.

\begin{equation} \label{Eq:Aux2}
    c(G_i) \leq \sum_{j \in G_i} c_j(\{j\})
\end{equation}

Eq. \eqref{Eq:Aux2} entails directly that $\Delta c(G_i) \leq 0$ (due to Eq. \ref{Eq:Residual}), which implies that $p(j,G_i) \leq 0 \: \forall j \in G_i$ (Eq. \ref{Eq:Sign}), and thus $c_j(G_i) \leq c_j(\{j\})$, which is the definition of a TNE.

The proof that no two groups can be mergeable follows directly from Lemma \ref{Lemma:BudgetBalancedRUE}.

\end{proof}

For the numerical simulations, we will use residual prices that are proportional to the individual costs, i.e.

\begin{equation}
    p(i,G)=\Delta c(G) \cdot \frac{c\{i\}}{\sum_{j \in G} c_j(\{j\})}
\end{equation}

Which implies that total costs for each user will be proportional to their private costs.

\subsection{Subgroup-based protocol}
A different way to combine individual and group costs is based on looking at the subgroups each player might belong to. Consider a group $G$ and a subgroup $H \subseteq G$. If $H$ presents a low average cost (in comparison with $G$), then players that belong in $H$ should not be charged much, as they are not directly inducing the costs on $G$.  On the other hand, consider a player $i$ such that every subset $H \subset G$ containing $i$ presents a higher average cost than $G$: in that case, player $i$ would be directly benefited when $G$ takes place, so she should pay more than the average. With this idea in mind, let us propose algorithm\footnote{We choose this writing of the algorithm because it eases the understanding. The algorithm can be adapted to run in polynomial time: Sort in a list all groups in $\mathcal{G}$ according to their average costs. Then, for each $G$, go through the said list, selecting the groups that are subsets of $G$ and that only contain elements that have not been assigned yet.} \ref{Alg:Subgroup} that takes as input the group $G$ and outputs, for each player $i$, a subgroup-based cost $z_i(G)$ and an associated subgroup $\varphi_i(G)$.

\begin{algorithm}
\begin{algorithmic}
\caption{Determining subgroup-based costs and the associated subgroups.}\label{Alg:Subgroup}
\STATE{Input: $G$}
\STATE{$W=G$; \% $W$ contains the users whose costs have not been determined yet}
\WHILE{$W\neq \emptyset$}
\STATE {$J = \argmin_{H' \subseteq W} \frac{c(H')}{|H'|}$; }
\FORALL{$i \in J$}
\STATE{$z_i(G)\leftarrow\frac{c(J)}{|J|}, \varphi_i(G)\leftarrow J$;}
\ENDFOR
\STATE{$W \leftarrow W \setminus J$;}
\ENDWHILE
\STATE{Output: $z_i(G), \varphi_i(G), \forall i \in G$}
\end{algorithmic}
\end{algorithm}

The associated subgroups $\varphi_i(G)$ represent which group is defining the cost of each $i$. It is relevant to remark a straightforward property fulfilled by these subgroups: they form a partition of $G$, that is to say, two conditions are met:
\begin{enumerate}
    \item $\forall i,j \in G$, either $\varphi_i(G) \cap \varphi_j(G)=\emptyset$ or $\varphi_i(G) = \varphi_j(G)$.
    \item $\bigcup_{i \in G} \varphi_i(G)=G$.
\end{enumerate}

Utilising costs $z_i$ would not necessarily induce a budget-balanced protocol. Let us define the excess $e(G)$ as:

\begin{equation}
    e(G)=c(G) - \sum_{i \in G} z_i(G)
\end{equation}

This excess might be positive or negative. To define the actual protocol, we modify the payoffs $z_i$ to achieve the budget-balanced property. If the excess is positive, we split the remainder uniformly among the users of the group (for our results below regarding optimal matching being equilibria, this remainder could be split in any way):

\begin{equation}\label{Eq:CostsSubgroupPositiveExcess}
  c_i(G) = z_i(G) + \frac{e(G)}{|G|} \: \text{ if } e(G) >0
\end{equation}

When $e(G)<0$, we sort the $\varphi_i(G)$ according to their average costs (from the cheapest to the most expensive), resulting in $\varphi^1,...,\varphi^Q$, with $Q$ the number of distinct subgroups $\varphi_i(G)$. Let us denote the respective costs $z^i=c(\varphi^i)/|\varphi^i|$ (we are omitting the explicit reference to $G$ to ease the notation). Define $I_1$ as the smallest index such that $z^{I_1} \geq c(G)/|G|$ (such an index exists because the excess is negative). We then diminish the costs of $\varphi^j$ for $j \geq {I_1}$ to be equal to the average cost of $G$ until we reach zero excess. The index $I_2$ representing the last $\varphi^j$ whose cost is adjusted to $c(G)/|G|$ can be characterised as the largest index such that

\begin{equation} \label{Eq:I2}
 \sum_{i=1}^{I_1-1} z^i |\varphi^i| + \sum_{i=I_1}^{I_2} \frac{c(G)}{|G|} |\varphi^i| + \sum_{i=I_2+1} ^Q z^i |\varphi^i| \geq c(G) 
\end{equation}

The definition of $I_2$ through Eq. \eqref{Eq:I2} implies that if we set the costs for users in $\varphi^{I_2+1}$ to be equal to $c(G)/G$, then the individual costs would not be enough to cover the group's costs. However, such costs might still be reduced from the original $z^{I_2+1}$, till the zero excess is reached. The precise expression for the budget-balanced costs in the case of negative excess is defined by:

\begin{equation} \label{Eq:CostsSubgroupNegativeExcess}
    \forall k \in \varphi^i, c_k(G)=
    \begin{cases}
    z^i & \text{ if } i<I_1 \text{ or } i>I_2+1 \\
    \frac{c(G)}{|G|} & \text{ if } i \in \{I_1,\ldots,I_2\} \\
    \frac{c(G) - \left(\sum_{j=1}^{I_1-1} z^j|\varphi^j| + \sum_{j=I_1}^{I_2} \frac{c(G)}{|G|}|\varphi^j| + \sum_{j=I_2+2} ^Q z^j|\varphi^j|\right)}{|\varphi^i|} & \text{ if } i=I_2+1
    \end{cases}
\end{equation}

The last case in Eq. \eqref{Eq:CostsSubgroupNegativeExcess} represents how to diminish the costs for users in $\varphi^{I_2+1}$ to make the excess equal to zero, which is not enough for $c_k(G)$ to reach $c(G)/|G|$ (by definition of $I_2$). Note that, in this case of negative excess, we only diminish some costs, which implies that
\begin{equation} \label{Eq:ciLowerThanzi}
c_i(G) \leq z_i(G) \forall i \in P 
\end{equation}

\begin{definition}
Let $G$ be a group. The \textbf{subgroup-based protocol} charges to each $i$ in $G$ the costs according to Eq. \eqref{Eq:CostsSubgroupPositiveExcess} if $e(G)\geq 0$, or to Eq. \eqref{Eq:CostsSubgroupNegativeExcess} if $e(G)\leq 0$.
\end{definition}

\begin{theorem} \label{Thm:Subgroups}
Assume that the subgroup-based protocol is applied. Then every optimal matching is both a RHE and a RUE.
\end{theorem}

\begin{proof}
In the appendix.

\end{proof}

\subsection{Synthesis and analysis}
We have proposed three cost-sharing protocols, that determine how to split the costs among users sharing a ride. All these protocols depend only on the ride itself, without requiring any more information, and two of those are budget-balanced.  Each equilibrium notion from section \ref{scn:FormalGame}, has at least one corresponding protocol that enables the system to simultaneously reach optimum and equilibrium. Table \ref{tab:SynthesisEqPro} synthesises which protocol should be used depending on which equilibrium notion governs the system:

\begin{table}[H]
    \centering
    \begin{tabular}{|c|c|c|c|}
    \hline
     \textbf{Equilibrium}    & RHE & RUE & RSIE  \\
     \hline
    \textbf{Protocol to be used} &  Subgroup-based & Subgroup-based or Residual-based & Externality-based\\
    \hline
    \end{tabular}
    \caption{Synthesis of which protocol reaches PoS=1 depending on the equilibrium notion.}
    \label{tab:SynthesisEqPro}
\end{table}

Regarding the protocols, the rationale behind each of them can also be synthesised: the externality-based accounts for the extra costs induced to the other members of the group, the residual-based rests mostly on the individual costs (e.g., proportional to the time required to go from the origin to the destination), whereas the subgroup-based recognises which users could be efficiently matched together in smaller groups to determine how to split the costs.

We can gain more intuition by moving back to the illustrative example. We exhibit the resulting costs when applying each protocol in Table \ref{tab:ExampleIntro2}.

\begin{table}[H]
    \centering
    \begin{tabular}{|c|c|c|c|c|c|c|c|c|c|c|}
    \hline
    \multirow{2}{*}{\textbf{Group}}     & \multirow{2}{*}{\textbf{Total cost}} & \multicolumn{3}{c|}{\textbf{Externality-Based}} & \multicolumn{3}{c|}{\textbf{Residual-Based}} & \multicolumn{3}{c|}{\textbf{Subgroup-Based}}  \\
    \cline{3-11}
    & & User 1 & User 2 & User 3 & User 1 & User 2 & User 3 & User 1 & User 2 & User 3 \\
    \hline
    $\{1\}$ & 23 & 23 & - & - & 23 & - & -& 23 & - & - \\
    $\{2\}$ & 19 & - & 19 & - & - & 19 & - & - & 19 & - \\
    \rowcolor{Gray} 
    $\{3\}$ & 9.24 & - & - & 9.24 & - & - & 9.24 & - & - & 9.24 \\
    \rowcolor{Gray}
    $\{1,2\}$ & 31 & 12 & 8 & - & 16.98 & 14.02 & - & 15.5 & 15.5 & - \\
    $\{1,3\}$ & 31.96 & 22.72 & - & 8.96 & 22.8 & - & 9.16 & 22.72 & - & 9.24 \\
    $\{2,3\}$ & 26.08 & - & 16.84 & 7.08 & - & 17.55 & 8.53 & - & 16.84 & 9.24 \\
    $\{1,2,3\}$ & 40.44 & 14.36 & 8.48 & 9.44 & 18.15 & 15 & 7.29 & 15.6 & 15.6 & 9.24 \\
    \hline
    \end{tabular}
    \caption{Resulting costs per user for the Example defined by Figure \ref{Img:ExampleIntro} and Table \ref{tab:ExampleIntro}, depending on the protocol applied. Grey rows represent the matching that minimises the sum of total groups' costs.}
    \label{tab:ExampleIntro2}
\end{table}

Let us analyse in more detail the resulting costs for the first two passengers depending on the protocol. In this case, the direct users' costs are $8$ for the first user (only in-vehicle time) and also 8 for the second user ($6$ in-vehicle and $2$ waiting), while the remaining $15$ are operators' costs; how to split such operators' costs can be seen as the fare paid by the users.
\begin{itemize}
    \item Externality-based: The costs charged to the second passenger are equal to her direct users' costs (waiting and in-vehicle times), meaning that she would be charged no monetary fee, which happens because she induces no detour to the vehicle that would be already transporting the first passenger. The first passenger, on the other hand, does need to pay more than the direct users' costs, because sharing the vehicle implies that the second user needs to wait, which wouldn't happen if she would have travelled alone. Note that the monetary fee for the first passenger would be $4$, so the group would require to be subsidised with $11$.
    \item Residual-based: User 1 is being charged more than user 2 because she is travelling a longer distance. Thus the residual cost (Eq. \ref{Eq:piG}), distributed proportionally to the costs of individual rides, affects the first traveller more.
    \item Subgroup-based: As the optimal subgroup for the two users is the same, they are both being charged the same total amount. As they both have equal direct users' costs, the monetary fee would also be equal ($7.5$ each).
  
\end{itemize}
It is worth recalling that everybody travelling alone forms a TNE, i.e., users 1 and 2 could be split although that would be inconvenient for both of them, underpinning the idea that TNE is too weak.



\section{Numerical simulations}\label{scn:results}

\subsection{The scenario}

We consider a batch of 400 travellers departing within a 10-minute window (2400 passengers/h) in the PM peak-hour in Amsterdam. Trips, defined through the exact origin, destination and departure time, are sampled from the nation-wide demand dataset (\cite{arentze2004learning}). Travellers have their individual costs $C(G,i)$ covering their in-vehicle travel and waiting times (denoted as $t(G,i)$ and $w(G,i)$, respectively), and they share vehicle's costs with fellow travellers.
We compute travel costs using average value of time (9 \EUR{}/hour) as $\beta_t$, and we multiply it by 1.5 to get the penalty for waiting $\beta_w$. We assume here that travel time is weighted equally regardless of the number of co-travellers:
\begin{equation}
    C(G,i) = \beta_t t(G,i) + \beta_w w(G,i)
\end{equation}
We calculate the vehicle costs $C_O(G)$ proportional to the trip distance $l(G)$ plus a fixed ride cost:
\begin{equation}
    C_O(G) = \beta_l l(G) + \beta_V,
\end{equation}
 with $\beta_l = 1$ \EUR{}/KM and $\beta_V=1$\EUR{}.
We do not consider extra societal costs in the experiments. Hence, total group costs $c(G)$ are composed of $\sum_{i \in G}C(G,i)$ and operator's costs, $C_O(G)$ only.

To compute the set of feasible groups $\mathcal{G}$ we apply the hierarchical exact algorithm of ExMAS (\cite{kucharski2020exact}) which computes all feasible groups of travellers (of any size). A group is declared feasible when for all the travellers the additional detour and delay can be compensated thanks to sharing. ExMAS in this configuration generates about 2200 feasible groups of various sizes, constituting $\mathcal{G}$ for further calculations.

We first calculate the pricings in the three proposed protocols for each traveller-ride combination. Then we prune the groups according to the different notions of equilibria, and finally we perform the matching to assign travellers to groups. Matching is done first with the objective to minimise total travel costs (to compute the price of stability of the system) and then to maximise it (to compute potential anarchy of the system).

\subsection{Results}

First we illustrate how the different equilibrium notions restrict the number of feasible solutions, by means of the so-called `shareability graph' (Figure \ref{fig:grafs}). In such a graph, each traveller is a node, and nodes $i$ and $j$ are connected if the corresponding users could be feasibly matched together (that is, if the group $\{i,j\} \in \mathcal{G}$). In Figure \ref{fig:grafs} we show how TNE and RHE pruned the initial $\mathcal{G}$ (recall that the other notions of equilibrium - RUE and RSIE - prune solutions based on exclusive pairs of groups rather than unfeasible groups). Out of 2191 initially feasible groups, 1708 remain for TNE and 1366 for RHE. However, such pruning does not significantly alter the graph structure, that always consists of one highly connected giant component and few disconnected nodes.

Table \ref{tab:prunings} summarises the results of the simulation of the effects of the different equilibrium notions, where we report how the 2191 initially feasible groups are pruned and how it affects the size of the remaining groups. TNE prunes mainly groups of size 2, with groups shared by more than two travellers remaining almost intact. Consequently, filtering 2191 groups to 1708 for TNE decreases the mean degree mildly. When we allow subgroups to coordinate to leave together, i.e. when we study RHE, rides of greater size are further pruned, resulting in a significantly lower average degree.

We also report in Table \ref{tab:prunings} the number of mutually exclusive constrains imposed by the RUE and RSIE protocols. This number is much higher for RSIE, i.e., when users can move individually to other group willing to receive her, the resulting equilibria seem to be much more restrictive than when admitting merges between groups. This fits the fact that a RUE always exists, whereas a RSIE does not necessarily exist.

\begin{figure}[H]
     \centering
     \begin{subfigure}[b]{0.3\textwidth}
         \centering
         \includegraphics[width=\textwidth]{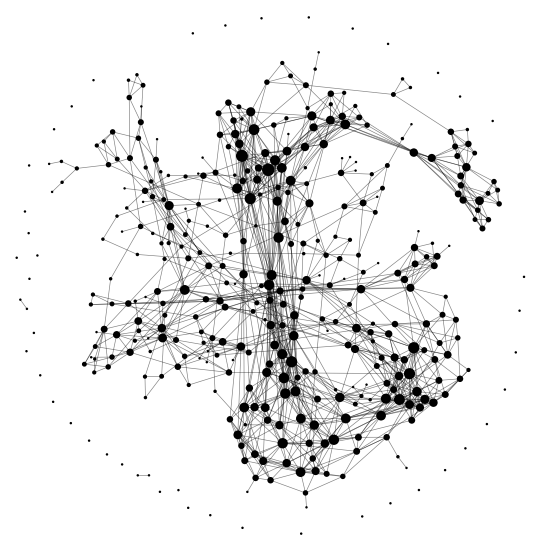}
         \caption{basic (2191 groups)}
         \label{fig:graf_a}
     \end{subfigure}
     \hfill
     \begin{subfigure}[b]{0.3\textwidth}
         \centering
         \includegraphics[width=\textwidth]{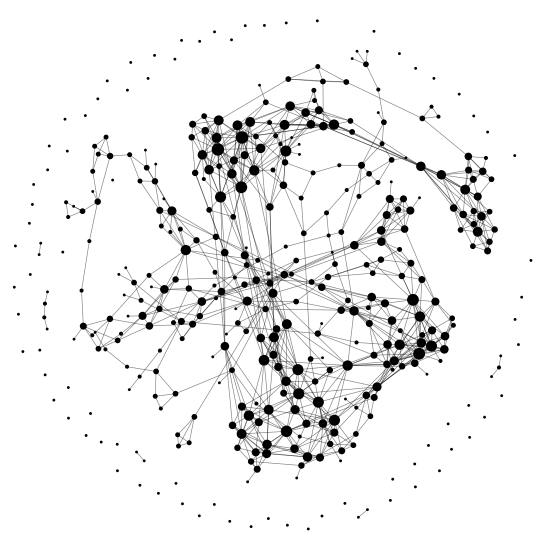}
         \caption{TNE (1708 groups)}
         \label{fig:graf_b}
     \end{subfigure}
     \hfill
     \begin{subfigure}[b]{0.3\textwidth}
         \centering
         \includegraphics[width=\textwidth]{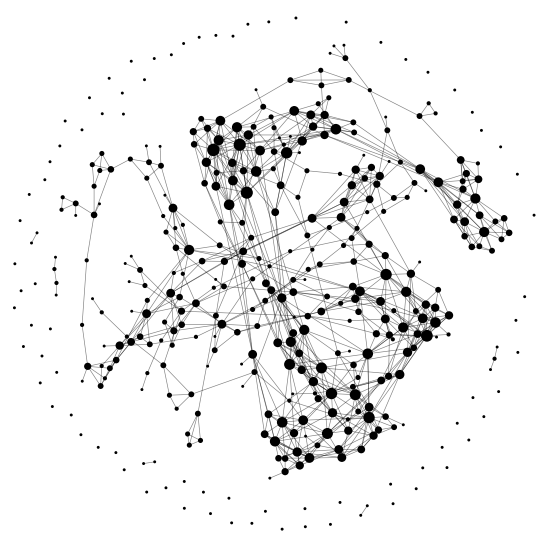}
         \caption{RHE (1366 groups)}
         \label{fig:graf_c}
     \end{subfigure}
        \caption{\small{Shareability graphs for selected pruning algorithms. Nodes denote travellers, which are linked if they can share a ride. Nodes are sized according to their degree (i.e. number of connecting edges). The number of nodes remains fixed across the pruning, and the number of links decreases as groups are being excluded in the pruning procedures.}}
        \label{fig:grafs}
\end{figure}

\begin{table}[ht]
\centering
\resizebox{0.5\textwidth}{!}{
\begin{tabular}{l|r|rrrr|r|r}
\toprule
{algorithm} &  number &   of size    &        &       &  & mutually   &    mean \\
{} &  of groups &      1 &       2 &      3 &     4+ &  exclusives & degree   \\
\midrule
Basic    &              2191 &  400 &  1348 &  363 &  79 &                   - &  2.05 \\
TNE      &              1708 &  400 &   928 &  305 &  74 &                   - &  2.03 \\
HERMETIC &              1366 &  400 &   806 &  142 &  18 &                   - &  1.83 \\
RUE      &              2191 &  400 &  1348 &  363 &  79 &                2762 &  2.05 \\
RSIE     &              2191 &  400 &  1348 &  363 &  79 &               24404 &  2.05 \\
\bottomrule
\end{tabular}}
\caption{The number of groups and the distribution of their degrees yielded by each of the pruning algorithms, as well as the number of mutually exclusive constrains (between group pairs) in the respective equilibrium notions.} 
\label{tab:prunings}
\end{table}

We now analyse the different cost-sharing protocols. Intuitively, one expects that shorter trips result in lower costs, as well as groups of larger sizes (because the costs are split among many). This is studied in Figure \ref{fig:Scatter}, where we show users' costs for the 400 travellers in 2191 groups (4085 traveller-group pairs), and their relationship with distance and degree, for the different cost-sharing protocols; we also compare these results with users' direct costs $c(G,i)$, i.e., when accounting only for their travelling times. 

The top row of Figure \ref{fig:Scatter} shows the relationship between costs and distance. The colour of the dots represents the size of the respective group. We  see  that  direct  costs (left  panel) are always lower for single rides, because there is no detour. However, when a cost-sharing protocol is introduced, that is, when the operators' costs are split among the users, larger groups become more attractive and mostly dominate private rides. This conclusion is reinforced when examining the bottom row, that shows users' costs depending on the groups' size, where it is apparent that the increasing trend observed in the left panel is neutralised or reversed in the other panels. There are some noteworthy aspects for each protocol:

\begin{itemize}
    \item \textbf{Subgroup-based}: This protocol does not exhibit a high correlation between distance and users' costs, as many dots are placed far from the diagonal. In addition, it is the protocol in which large groups are favoured the most, as most red dots present the lowest cost for a given distance. This is in line with intuition, as this protocol is meant to ensure that efficient groups are hermetic, so that (a number of) their members are assigned with the lowest costs when the group is complete.
    \item \textbf{Residual-based}: This protocol exhibits the highest correlation between distance and users' costs, which is expected since costs are calculated proportionally to distances.
    \item \textbf{Externality-based}: This protocol exhibits the lowest correlation between distance and users' costs. This means that the protocol is effectively capturing that the relevant aspect for this protocol is not the total distance, but the total detour imposed on others. The few dots at the bottom with zero costs represent cases in which a group $H$ would not be feasible according to the ExMAS algorithm which determines the set $\mathcal{G}$, but there is some feasible group $G$ containing $H$. In those cases, $H$ is added assuming the cost of the cheapest feasible group containing it.
    
\end{itemize}

\begin{figure}[H]
    \centering
    \includegraphics[width=\textwidth]{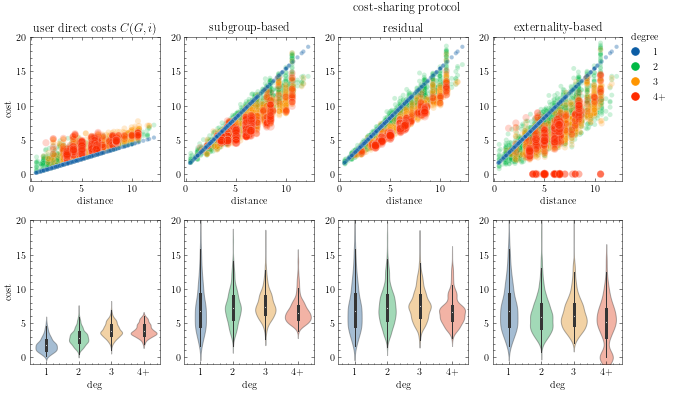}
    \caption{\small{User costs as a function of distance (top row). First we plot direct user costs, followed by the costs under the three introduced cost-sharing protocols. Each dot denotes a traveller-group pair and colours denote the size of the corresponding group. The blue diagonal line denotes non-pooled rides. The bottom row shows the cost distributions varying with the groups' size for the different protocols.}}
    \label{fig:Scatter}
\end{figure}

In Table \ref{tab:GlobalResultsPoSPoA} we report the most relevant indices for each protocol. We study the case in which the operator can propose a solution (`Best case', the one that we focus on throughout the paper), but we also analyse the `Worst case' where we assume the users are not coordinated, so they may reach the solution of maximal, rather than minimal, system-wide costs. 

We report passenger-hours, vehicle-hours, number of groups and total cost (reported as PoS - best and PoA - worst, respectively). The best-case analysis confirms that for each equilibrium there is at least one pricing allowing to reach the PoS$=1$ (as synthesised in Table \ref{tab:SynthesisEqPro} in the previous section). Moreover, for all the protocols, the price of stability is low regardless of the equilibrium notion, meaning that the protocols are robust. In fact, all the indices are very similar across the protocols, i.e., they yield similar best solutions. 

The worst-case results show that the price of anarchy is between 1.1 and 1.22 in this example. That is to say, if the system does not propose a solution, losses can be as high as 22\%, which highlights the relevance of having centralised solutions. The residual-based protocol is the most effective one in the case with no central coordination.

Note that the highest PoA is always reached when studying RHE, because everybody travelling alone is a RHE, which is exactly the worst solution (400 groups). For RUE and RSIE the contrary happens, i.e., passenger-hours are similar (albeit higher) to those obtained in the best-case analysis, while increasing vehicle-hours and the number of groups. That is to say, while the system is providing almost the same quality of service as in the best-case scenario, it does so through a non-efficient utilisation of the vehicles. In the case of RUE, this happens because this notion of equilibrium can admit any large group, regardless of its efficiency; something similar happens with RSIE, but indirectly: when groups are large, few individual movements from one group to another are feasible.

In Table \ref{tab:Subsidies} we report which portion of the total costs are covered by the sum of users' individual costs in the best-case solution for the only protocol which is not budget-balanced: externality-based protocol. Even though more than 80\% of the total costs are covered by the sum of users' costs, additional subsidies would be needed regardless of the equilbrium notion.


\begin{table}[H]
    \centering
    \resizebox{0.9\textwidth}{!}{
    \begin{tabular}{|c|c|c|c|c|c|c|c|c|c|}
    \hline
         &  & \multicolumn{4}{c|}{\textbf{Best-Case}} & \multicolumn{4}{c|}{\textbf{Worst-Case}} \\
         \cline{3-10}
    \textbf{Protocol} & \textbf{Equilibrium} & \textbf{Pax-Hours} & \textbf{Veh-Hours} & \textbf{Nº of groups} & \textbf{PoS} & \textbf{Pax-Hours} & \textbf{Veh-Hours} & \textbf{Nº of groups} & \textbf{PoA} \\
    \hline
    \multirow{2}{*}{Residual-based} & RHE & 83.64 & 41.84 & 234 & 1.001 & 60.34 & 60.34 & 400 & 1.22 \\
     & RUE & 84.01 & 41.64 & 234 & 1 & 89.31 & 54.46 & 275 & 1.18 \\
      & RSIE & 84.12 & 41.66 & 234 & 1.001 & 93.28 & 49.08 & 235 & 1.12 \\
      \hline

    \multirow{2}{*}{Subgroup-based} & RHE & 84.01 & 41.64 & 234 & 1 & 60.34 & 60.34 & 400 & 1.22 \\
     & RUE & 84.01 & 41.64 & 234 & 1 & 88.82 & 54.57 & 277 & 1.19 \\
      & RSIE & 83.16 & 42.25 & 235 & 1.003 & 93.4 & 49.53 & 234 & 1.12 \\
      \hline
      \multirow{2}{*}{Externality-based} & RHE & 82.14 & 42.62 & 240 & 1.006 & 60.34 & 60.34 & 400 & 1.22 \\
       & RUE & 84.01 & 41.64 & 234 & 1 & 89.02 & 54.59 & 276 & 1.19 \\
        & RSIE & 84.01 & 41.64 & 234 & 1 & 92.88 & 50.11 & 238 & 1.13 \\
    
    \hline
    \end{tabular}}
    \caption{KPIs for the three cost-sharing protocols proposed in the paper, depending on the equilibrium notion that governs the system. On the left side we show the `best-case' results, i.e. when we minimise total costs subject to the respective equilibrium, representing the best solution that could be proposed by the system's operator. On the right side we show the `worst-case' solution, obtained by maximising costs subject to the respective equilibrium, representing the worst possible outcome if users coordinate by themselves.}
    \label{tab:GlobalResultsPoSPoA}
\end{table}

\begin{table}[H]
    \centering
     \resizebox{0.3\textwidth}{!}{
    \begin{tabular}{|c|c|}
    \hline
    \textbf{Equilibrium}     & \textbf{Portion}  \\
    \hline
    RHE     & 0.81 \\
    RUE & 0.82 \\
    RSIE & 0.82 \\
    \hline
    \end{tabular}}
    \caption{Portion of the total costs covered by the sum of users' costs when the externality-based protocol is utilised, considering the best-case solution.}
    \label{tab:Subsidies}
\end{table}

Finally, in Figure \ref{fig:hist} we study users' satisfaction in comparison to the cheapest alternative they have (i.e., if they had the opportunity to select their most preferred group, regardless of their co-travellers' opinion). To do this, we plot the cumulative distribution of the relative differences between the cheapest alternative for each user, and the individual cost of the group that actually contains her in the best possible matching for each equilibrium notion. Null difference implies a perfect match (best personal option) which increases as users are matched to more expensive pooled rides. These curves complement system-wide indicators, as they allow to investigate distributional effects among users. The  faster the curve reaches 1, the more equalised the matching.   

We observe that the panels in Figure \ref{fig:hist} look similar regardless of the equilibrium notion. For all of them, the externality-based protocol provides the least equal outcome, with some users facing costs that can be higher than 2 times their best option, whereas using the subgroup-based or residual-based protocols this number is reduced to little more than 1.5. The residual-based protocol achieves the highest equity. Recalling that RSIE charges less from the users than the other protocols (as it does not collect enough to cover the whole operators' costs), an interesting trade-off emerges between total users' costs and its (unequal) distribution.

\begin{figure}[H]
    \centering
    \includegraphics[width=\textwidth]{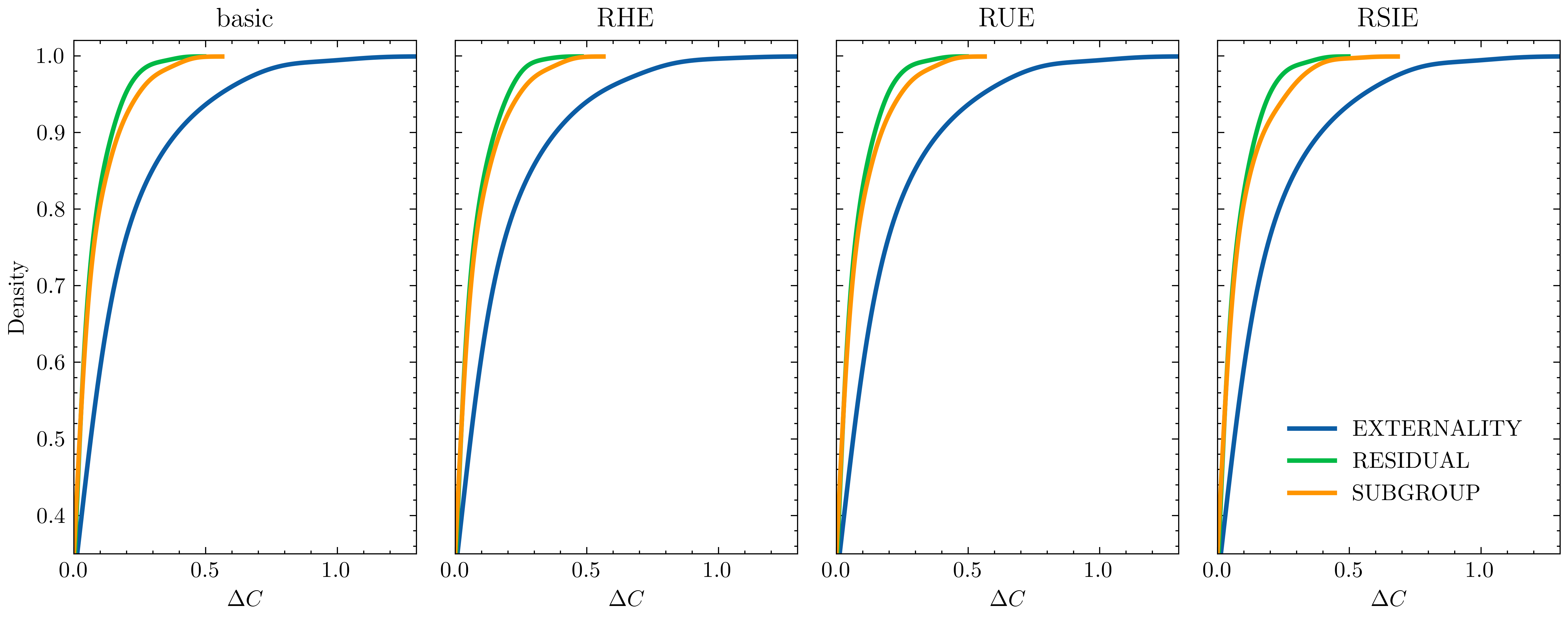}
    \caption{\small{Cumulative distribution of relative differences between each user's lowest possible cost, and the cost imposed on her by the best possible matching for each equilibrium notion. Different colors represent the various cost-sharing protocols.}}
    \label{fig:hist}
\end{figure}

\section{Synthesis, conclusions and future research}\label{scn:conclusions}

In this paper, we address the issue of how to split common costs when users share a ride in a mobility system that decides how to group the users. After recognising that the way costs are split can affect which groups are going to be formed and hence the quality of the solution (matching), we have modelled the said situation as a game, in which each user can choose with whom to travel, as long as all co-travellers agree.

In order to study the possible equilibria in such a game, we have proved that it suffices to consider pure strategies only. Moreover, we have discussed that the traditional notions of Nash and Strong Equilibria are not the most appropriate ones as they may either prevent sharing on one hand (Nash) or require unfeasible coordination on the other (Strong). We therefore proposed three intermediate definitions of equilibrium depending on which are the possible ways in which users can coordinate. Moreover, we have also proved that determining whether a Strong Equilibrium exists is an NP-Complete problem, through a reduction from the so-called $3$-Exact-Cover.

For each of these equilibrium notions we have proposed a corresponding cost-sharing mechanism that reaches a price-of-stability equal to one, i.e., that makes any optimal solution also an equilibrium. By this means, we allow a central operator to group the users, simultaneously reaching a system-wide optimum and a users' equilibium. When deciding what is the cost for an individual within a given group, the protocols determine based on (i) the cost of the respective private trip (`residual-based'), (ii) the extra-costs imposed to the co-travellers (`externality-based'), or (iii) the costs of the subgroups that contain the said user (`subgroup-based').

We tested our ideas by computing the feasible groups for a set of 400 travellers in Amsterdam forming 2200 feasible groups. Numerical results show that our protocols effectively make efficient larger groups to be preferable by the users, and that they always reach a price-of-stability close to 1. We also show that if the system cannot propose a solution and users coordinate by themselves then the worst case induces around 20\% extra-costs for two of the equilibrium notions.

Our discussion and methods demonstrate that selecting an appropriate cost-sharing mechanism, understandable by the users, can play a key role in ensuring that an equilibrium exists and in aligning users' interests with a system-wide optimum. Furthermore, the ability of proposing a central coordination is also crucial to make this type of mobility systems attractive.

As this is an emerging topic, there are numerous directions for further research. In this paper, we have assumed a central operator that aims for a global optimum; if a for-profit company was considered instead, our model would have to be modified to consider the company's interests as well (for instance, as an additional player). We assumed the demand is exogenous and fixed, whereas in fact such cost-sharing protocols can make the system more attractive and thus induce the demand, triggering a positive feedback loop as the critical mass needed for pooling is reached. Moreover, we have assumed that the system's optimum depends only on total costs, regardless of equity aspects, although we show that results can actually be far from equal for all the users involved, which suggests yet another direction for further research. On a theoretical note, studying protocols that ensure bounds on the price of anarchy, as well as determining the complexity of deciding whether a RSIE exists if all the groups have size no larger than 2, are relevant questions that emerge from this paper. Finally, equilibrium analysis in a dynamic environment, taking into account users that will emerge in the future for which there is partial or no information available, is also an intriguing research avenue, that has been analysed only from the point of view optimal matching, both theoretically (\cite{Feng2020Two,torrico2017dynamic}) and in applied models (\cite{wallar2018vehicle,wen2017rebalancing,van2018enhancing,alonso-mora_predictive_2017}).

\section*{Acknowledgments}
The authors want to thank Jos\'e Correa, from Universidad de Chile, and Orlando Rivera, from Universidad Adolfo Ib\'a\~{n}ez, for the fruitful discussions that helped to elaborate the contents of this paper. The work of the second and third authors was supported by the CriticalMaaS project (grant number 804469), which is financed by the European Research Council and the Amsterdam Institute for Advanced Metropolitan Solutions.

\printbibliography

\section*{Appendix}

\subsection*{Continuation of the proof of Theorem \ref{Thm:TSENPComplete} a)}

We now need to prove that there is an exact cover in the original instance of $3$-exact-cover iff there is a strong equilibrium in this instance of $CTG$.

If there is an exact cover $b_1,...,b_p$ in the original instance, it is straightforward how to build a strong equilibrium in the instance of $CTG$. Indeed, consider the profile of strategies in which $\alpha$ travels alone, and the other groups are $b_1,...,b_p$. Therefore, every element of $S$ is in its most preferred situation, and $\alpha$ has no other choice than being alone.

We will now show that these are the only possible strong equilibria in such an instance of $CTG$, more precisely we will prove the following Lemma:

\begin{lemma} \label{Lemma1}
Let $h_1,...,h_q$ be a strong equilibrium on the proposed instance of $CTG$. Then there exists $\ell$ such that $h_\ell=\{\alpha\}$ and $\forall j \neq i, h_j$ corresponds to some $b_i$ from $\mathcal{B}$.
\end{lemma}

Lemma \ref{Lemma1} implies that when there is a strong equilibrium in the instance of $CTG$, we can construct the exact cover just by taking all such subsets $h_j$. The proof of Lemma \ref{Lemma1} rests on another technical Lemma, that limits which situations can occur in a strong equilibrium:

\begin{lemma}\label{Lemma2}
Let $h_1,...,h_q$ be a strong equilibrium on the proposed instance of $CTG$. Then none of the following situations can happen:
\begin{enumerate}[i)]
    \item The existence of one $h_i$ of size 3 that does not come from $\mathcal{B}$ and the existence of one $h_j$ of size one.
    \item The existence of two $h_i,h_j$ both of size 2.
    \item The existence of two $h_i,h_j$ both of size 1.
\end{enumerate}

\end{lemma}

Before proving Lemma \ref{Lemma2}, let us explain why it suffices to prove Lemma \ref{Lemma1}. In fact, the number of players in this game is $3t+1$, for some $t \in \mathbb{N}$. If we remove all those players that belong to a subset $h_i$ coming from $\mathcal{B}$, the remaining number of players is $3s+1$ for some $s \in \mathbb{N}$. We just need to show that $s=0$, i.e., that all players but $\alpha$ are covered by groups coming from $\mathcal{B}$. What Lemma \ref{Lemma2} ensures is that if $s>0$, then it would be impossible to cover the remaining $3s+1$ players with the subsets $h_i$ not coming from $\mathcal{B}$, which is a contradiction. Hence, Lemma \ref{Lemma1} is proven by Lemma \ref{Lemma2}.

Putting everything together, we just need to prove that Lemma \ref{Lemma2} is true. We will show that each of the forbidden situations is indeed impossible in a strong equilibrium, following the same order:

\begin{enumerate}[i)]
    \item Assume that there is a group $h_i=\{x,y,z\}$ that does not come from $\mathcal{B}$, and a group $h_j=\{w\}$. Without loss of generality, $w$ falls between $x$ and $y$ within the circle that defines the lexicographic order of the preferences. Hence, both $x$ and $w$ would be strictly better if they change to form the group $\{x,w\}$, contradicting the fact that this was a strong equilibrium.
    \item Assume two groups $h_i=\{x,y\}, h_j=\{w,z\}$. When we restrict the circle to these four players, either the preferred co-player for $x$ is not $y$, or the preferred co-player for $y$ is not $x$ (or both). Without loss of generality, we assume the former case. It is clear that the users of a group of size $2$ will always improve their situation if someone else joins, then $x,z$ and $w$ would coordinate to form the group $\{x,z,w\}$, contradicting the fact that this was a strong equilibrium.
    \item Two isolated players will always prefer to merge.
\end{enumerate}

I.e., the forbidden situations from Lemma \ref{Lemma2} cannot exist in a strong equilibrium, which completes the proof the Theorem.

$\qedsymbol$

\subsection*{Continuation of the proof of Theorem \ref{Thm:TSENPComplete} b)}

The stable-roommate problem ($SRP$) is defined by a set of $2 \mu$ players for some $\mu \in \mathbb{N}$, such that for each player there is a list sorting the rest of the players in some strict order of preferences. The problem consists on determining whether all the players can be grouped in pairs, in a way where it never happens that $i$ and $j$ are not together but they both would be better if they were. We now prove that $SRP$ is equivalent to the restricted version of $CTG$ in which all groups have size 1 or 2, that we denote $CTG-2$.

\begin{proof}
Let $\mathcal{G}$ be the set of groups in an instance of $CTG-2$. Consider a user $i \in P$. Note that her preferences can be described as a list $L_i=(y_1,\ldots,y_{k_i},i)$, meaning that her preferred group is $(i,y_1)$, followed by $(i,y_2)$, and so on until $(i,y_{k_i})$ and then to travel alone. The groups that come after travelling alone are not relevant as they shall never be part of any equilibrium. Utilising this notation, we build an instance of $SRP$ as follows:
\begin{itemize}
    \item The set of players is $P \times P'$, where $P'$ contains one copy of each player in $P$.
    \item If the list of preferences for $x \in P$ is $L_x=(y_1,...,y_{k_x},x)$ in $CTG-2$, the preferences in $SRP$ are $(y_1,...,y_k,x')$ (i.e., the same one but switching $x$ by its copy in $P'$); the order of the players after $x'$ is irrelevant for this proof.
    \item For $x' \in P'$, its preferred match would be $x \in P$. The list continues with the elements in $P'$ according to a circle (i.e., it prefers $y'$ over $z'$ iff $y'-x'$ mod $n < z'-x'$ mod $\mu$, as described in the proof of the part a) of this Theorem), and the rest of the preferences is irrelevant.
\end{itemize}
We need to prove that there is a strong equilibrium in the original instance of $CTG-2$ iff there is a stable matching in this instance of $SRP$. To do that, it is useful to note first that an even-sized circle $\{x_1,\ldots,x_{2\mu\}}$ always presents a stable matching, by joining $x_i$ with $x_{i+\mu (\text{ mod }\mu)}$.

Take first a strong equilibrium in $CTG-2$, formed by the pairs $(a_1,b_1),\ldots,(a_k,b_k)$ and the users $c_1,\ldots,c_\ell$ travelling alone. We build a stable matching in $SRP$ as follows:
\begin{itemize}
    \item $a_i$ and $b_i$ are matched in $P$ for all $i=1,\ldots,k$.
    \item $c_i$ in $P$ is matched with $c_i'$ in $P'$ for all $i=1,\ldots,\ell$.
    \item The only players that remain to be matched are the copies in $P'$ of $a_i$ and $b_i$, for all $i=1,\ldots,k$. Thus, it is an even number of players in $P'$, that form a circle, so we can create a stable matching among them as discussed above.
\end{itemize}
It is straightforward to see that this is indeed a stable matching, thanks to the fact that it was built from a strong equilibrium in $CTG-2$.

Take now a stable matching in $SRP$. We first note that in this stable matching there is no pair $(x,y')$, with $x \in P, y' \in P'$ and $x \neq y$: if there was such a pair, then the matching would not be stable because $x$ and $x'$ would prefer be matched together. Therefore, each $x \in P$ is either matched with a different $z \in P$ or with $x'$. We build the strong equilibrium taking exactly these assignments: each pair $(x,z)$ formed by two elements in $P$ will form a group in $CTG-2$, and those $y$ that are paired with their copy $y' \in P'$ will travel alone in $CTG-2$. It is straightforward to see that such an assignment makes a strong equilibrium, thanks to the fact that it was built from a stable matching.
\end{proof}

\subsection*{Proof of Theorem \ref{lemma:URE}}

We prove the Theorem by finding a RUE, which is done algorithmically. In short, we begin with everybody travelling alone, and each time we find two mergeable groups, we merge them, which we repeat until we find no more. The resulting profile of strategies is a RUE. We provide the respective pseudo-code in algorithm \ref{Alg:SRE}.

\begin{algorithm}
\begin{algorithmic}
\caption{Construction of a RUE.}\label{Alg:SRE}
\STATE {$\forall i \in P, G_i = \{i\}$};
\STATE{$v=0$}; \text{ \% \small{$v$ is an auxiliary variable to end the following while cycle.}}
\WHILE{v=0}
\STATE{$v=1$;}
\FORALL {$i,j \in P \text{ such that } G_i \neq G_j$}
\IF{$G_i$ and $G_j$ are mergeable}
\STATE{$v=0$;} \text{ \% The cycle continues}
\STATE{$G_i,G_j\leftarrow G_i \cup G_j$;}
\ENDIF
\ENDFOR
\ENDWHILE
\STATE{Output $(G_i)_{i \in P}$;}
\end{algorithmic}
\end{algorithm}

By construction, the output from algorithm \ref{Alg:SRE} is a RUE:
\begin{itemize}
    \item The algorithm only stops if there are no more pairs of mergeable groups.
    \item The output is also a TNE. As the algorithm only induces Pareto-improvements, everybody ends-up better-off (or equal) than in the initial situation, i.e., than travelling alone. This is exactly the definition of being a TNE.
\end{itemize}  
Moreover, algorithm \ref{Alg:SRE} does stop because it starts with a profile of strategies, and at each step that is kept but with groups of increasing size.

\subsection*{Example of an instance of CTG with no matching that is RHE and RUE}

Consider a game with five players $A,B,C,D,E$. As in the proof of Theorem \ref{Thm:TSENPComplete}, we do not need to impose that if $H \subseteq G$, then $c(H) \leq c(G)$, and it suffices to exhibit the players' orders of preferences. The set $\mathcal{G}$ is defined by:
\begin{equation*}
  \mathcal{G}=\{A,B,C,D,E,AB,AE,BC,CD,DE,ABC,ABE,ADE,BCD,CDE \}
\end{equation*}
The set $\mathcal{G}$ should have the property that if $G \in \mathcal{G}$ and $H \subseteq G$, then $H \in \mathcal{G}$, which is not the case with this definition of $\mathcal{G}$. We assume that those groups $H$ are in $\mathcal{G}$, but all members of $H$ would prefer to travel alone, so $H$ will not be a part of any equilibrium and we can omit them.

In Table \ref{Table:ExampleRUERHE} we show, for each user, how they sort their preferences. It is assumed that travelling alone comes right after the last element of the respective columns. For instance, the best choice for player $A$ would be to form the group $ABC$, whereas her worst option before travelling alone is $ADE$. Note that the preferences are symmetric: for player $p_i$ and using the sum mod $5$, her order of preferences is: $p_{i+1}p_{i+2},p_{i+1},p_{i+1}p_{i-1},p_{i-1},p_{i-2}p_{i-1}$.

\begin{table}[H]
\centering
\begin{tabular}{|L|L|L|L|L|}
\hline
\textbf{A} & \textbf{B} & \textbf{C} & \textbf{D} & \textbf{E}  \\
\hline
ABC         & BCD         & CDE         & DAE         & EAB          \\
AB          & BC          & CD          & DE          & EA           \\
ABE         & BAC         & CBD         & DCE         & EAD          \\
AE          & BA          & CB          & DC          & ED           \\
ADE         & BAE         & CAB         & DBC         & ECD \\    
A         & B        & C         & D         & E \\    \hline
\end{tabular}
\caption{Feasible co-travellers for each of the five users, sorted according to their preferences.}
\label{Table:ExampleRUERHE}
\end{table}

We now show that there is no matching for this instance that is a RUE and a RHE. First, such a matching could not contain any group of size $3$ because they are not hermetic. In fact, all those groups are of the type $p_ip_{i+1}p_{i+2}$, and the subgroup $p_{i+1}p_{i+2}$ would always want to leave. Discarding the groups of size $3$, there are three remaining options for $A$: travelling alone, with $B$ or with $E$. It is useful to note that two groups of the form $p_{i-1}p_i$ and $p_{i+1}$, respectively, are always mergeable.

\begin{itemize}
    \item If $A$ travels alone, $B$ can travel alone, but it would merge with $A$, or $B$ can travel with $C$. In the last case, $D$ would travel with $E$, but $A$ and $DE$ are mergeable.
    \item If $A$ travels with $B$, $C$ can travel alone or with $D$. If $C$ travels alone, $AB$ and $C$ are mergeable. If $C$ travels with $D$, $E$ travels alone, but $CD$ and $E$ are mergeable.
    \item If $A$ travels with $E$, $B$ can travel alone or with $C$. In the first case, $AE$ and $B$ are mergeable. In the second case, $BC$ and $D$ are mergeable.
\end{itemize}

Therefore, any feasible matching for this instance contains either a pair of mergeable groups or a non-hermetic group.

\section*{Proof of Theorem \ref{Thm:Subgroups}}
The proof consists of two parts. First, we prove that any group with negative excess is hermetic. Second, we show that in any optimal matching every group leads to a negative excess. The fact that every optimal matching is a RUE follows directly from Lemma \ref{Lemma:BudgetBalancedRUE}.

Let $G$ be a group with negative excess, implying that $c_i(G)\leq z_i(G)$ (Eq. \ref{Eq:ciLowerThanzi}), and let $H \subseteq G$. We shall show that users in $H$ do not want to coordinate to leave $G$. Take $i \in H$ such that 
\begin{equation}\label{Eqn:AuxSubgroups1}
\forall j \in H, z_i(G) \leq z_j(G)
\end{equation}

By definition of the functions $z$, Eq. \eqref{Eqn:AuxSubgroups1} implies that $\varphi_i(G)$ was the first one selected among $\{\varphi_j(G): j \in H\}$. Therefore, the set $\varphi_i(H)$ was completely available (i.e., in $W$ in Algorithm \ref{Alg:Subgroup}) when $\varphi_i(G)$ was selected. This implies that $z_i(H) \geq z_i(G)$, as $z_i(H)$ is computed optimising over a smaller set. If $e(H) \geq 0$, then 
\begin{equation} \label{Eq:ProofRHERUE}
   c_i(H) \geq z_i(H) \geq z_i(G) \geq c_i(G) 
\end{equation}
The positive excess of $H$ explains the first inequality, the second inequality was explained in the previous paragraph and the third inequality is due to the negative excess of $G$. If $e(H)<0$ and $c_i=z_i(H)$, Eq. \eqref{Eq:ProofRHERUE} still holds. Finally, if $e(H)<0$ and $c_i(H)<z_i(H)$:

\begin{equation}
    c_i(H) \geq \frac{c(H)}{|H|} \geq z_i(G) \geq c_i(G) 
\end{equation}

Where the first inequality is due to the definition of the costs in the case of negative excess, which diminishes the costs of some subgroups but never below the average cost of the whole group (Eq. \ref{Eq:CostsSubgroupNegativeExcess}), the second inequality is explained because $\varphi_i(G)$ is selected when the whole subset $H$ is available, and the third inequality holds because $e(G)\leq 0$ (Eq. \ref{Eq:ciLowerThanzi}). Putting everything together, we have shown that $c_i(H) \geq c_i(G)$, i.e., $i$ would not leave $G$ to form $H$. Therefore, $G$ is indeed hermetic.\\

\vspace{2mm}

The negative excess of $G$ in an optimal matching emerges when recalling that sets $\varphi^j$ are a partition of $G$. Indeed:
\begin{equation}
    \sum_{i \in G} z_i(G) = \sum_{i \in G} \frac{c(\varphi_i(G))}{|\varphi_i(G)|} = \sum_{j=1}^Q c(\varphi^j) \geq c(G)
\end{equation}

Note that the third sum is adding each $\varphi_i(G)$ once. The second equality is achieved by noting that each $\frac{c(\varphi_i(G))}{|\varphi_i(G)|}$ is added exactly $|\varphi_i(G)|$ times, and the final inequality is true because sets $\varphi^j$ are a partition, that has to be sub-optimal because forming the group $G$ is part of the optimal matching.

\end{document}